\documentclass[a4paper,UKenglish]{lipics-v2019}

\title{Dynamic Complexity Meets Parameterised Algorithms}

\author{Jonas Schmidt}{TU Dortmund University, Dortmund, Germany}{}{}{}%
\author{Thomas Schwentick}{TU Dortmund University, Dortmund, Germany}{}{}{}
\author{Nils Vortmeier}{TU Dortmund University, Dortmund, Germany}{}{}{}
\author{Thomas Zeume}{TU Dortmund University, Dortmund, Germany}{}{}{}
\author{Ioannis Kokkinis}{TU Dortmund University, Dortmund, Germany}{}{}{}

\authorrunning{J.\,Schmidt and T.\,Schwentick and N.\,Vortmeier and
  T.\,Zeume and I.\,Kokkinis}%

\Copyright{Jonas Schmidt and Thomas Schwentick and Nils Vortmeier and
  Thomas Zeume and Ioannis Kokkinis}%

\ccsdesc[500]{Theory of computation~Parameterized complexity and exact algorithms}
\ccsdesc[500]{Theory of computation~Logic and databases}
\ccsdesc[500]{Theory of computation~Complexity theory and logic}

\keywords{Dynamic complexity, parameterised complexity}%

\category{}%

\relatedversion{This is the full version of a paper to be presented at CSL 2020.}

\supplement{}%

\funding{The authors acknowledge the financial support by DFG grant SCHW 678/6-2.}%

\acknowledgements{We are grateful to Till Tantau for some valuable discussions.}%

\nolinenumbers %

\hideLIPIcs  %

\EventEditors{Maribel Fern\'{a}ndez and Anca Muscholl}
\EventNoEds{2}
\EventLongTitle{28th EACSL Annual Conference on Computer Science Logic (CSL 2020)}
\EventShortTitle{CSL 2020}
\EventAcronym{CSL}
\EventYear{2020}
\EventDate{January 13--16, 2020}
\EventLocation{Barcelona, Spain}
\EventLogo{}
\SeriesVolume{152}
\ArticleNo{36}

\usepackage{mathtools}
\usepackage{tcolorbox}
\usepackage{xspace}
\usepackage{ifmtarg}
\usepackage{soul}
\usepackage{tikz}
\usetikzlibrary{arrows, shapes, calc, decorations.pathmorphing}
\usetikzlibrary{positioning}
\usetikzlibrary{decorations.markings}

\usepackage{algorithm}

\usepackage{import}%

\usepackage[plain]{fancyref}

\fancyrefchangeprefix{\fancyrefseclabelprefix}{section}

\usepackage{newnotation}

\usepackage[]{apxproof} %

\theoremstyle{plain}
\newtheoremrep{theorem}{Theorem}
\newtheoremrep{proposition}[theorem]{Proposition}
\newtheoremrep{lemma}[theorem]{Lemma}

\newif\ifcomments
\newif\ifchanges
\commentsfalse\changesfalse

\notationnewgroup{ General Notation for Classes and Problems}
\newnotation{\myclass}[1]{\textsf{\upshape #1}}[Layout for all complexity classes][\notationarg{1}{Name}]

\newnotation{\StaClass}[1]{\myclass{#1}}[Layout for static complexity classes][\notationarg{1}{CName}]

\newnotation{\DynClass}[1]{\myclass{Dyn#1}}[ Layout for dynamic complexity classes (absolute and Delta-Semantics)][\notationarg{1}{Name}]

\newnotation{\ParaClass}[1]{\myclass{para-#1}}[Layout for parameterised complexity classes][\notationarg{1}{Name}]

\newnotationclass{myproblem}{\textsc}

\newnotation[myproblem]{\problem}[1]{#1}[Layout for static problems][\notationarg{1}{Name}]

\newnotation[myproblem]{\dynProb}[1]{Dyn(#1)}[Layout for dynamic problems][\notationarg{1}{Name}]
\newnotation{\paraproblem}[1]{\problem{$p$-#1}}

\newnotation{\class}{\calC}[An arbitrary complexity class]

\notationnewgroup{Traditional Complexity Classes}
\newnotation{\LOGSPACE}{\myclass{LOGSPACE}}
\newnotation{\NLOGSPACE}{\myclass{NLOGSPACE}}
\newnotation{\PSPACE}{\myclass{PSPACE}}
\newnotation{\PFP}{\myclass{PFP}}
\newnotation{\classL}{\myclass{L}}
\newnotation{\NL} {\myclass{NL}}
\renewnotation{\P}{\myclass{P}}
\newnotation   {\PTIME}    {\myclass{PTIME}}
\newnotation   {\NP}   {\myclass{NP}}
\newnotation   {\NC}   {\myclass{NC}}
\newnotation   {\AC}   {\myclass{AC}}
\newnotation     {\TC}   {\myclass{TC}}
\newnotation   {\SAC}   {\myclass{SAC}}
\newnotation   {\ACC}   {\myclass{ACC}}
\newnotation   {\tc}   {\myclass{TC}}   %
\newnotation   {\PPoly}{\myclass{\mbox{P}/\mbox{Poly}}} %
\newnotation   {\FOarb}   {\myclass{FO(arb)}}
\newnotation{\ACz}{\AC^0}
\newnotation{\ACzu}{\AC^{0\uparrow}}

\notationnewgroup{Logical Complexity Classes}
\newnotation{\FO}{\StaClass{FO}}
\newnotation{\FOar}{\StaClass{FO$(+,\!\times\!)$}}
\newnotation{\MSO}{\StaClass{MSO}}
\newnotation{\GSO}{\StaClass{GSO}}
\newnotation{\EMSO}{\StaClass{$\exists$MSO}}
\newnotation{\QFO}[1][\quant]{\StaClass{\ensuremath{#1}FO}}
\newnotation{\EFO}{\QFO[\exists^*]}
\newnotation{\AFO}{\QFO[\forall^*]}
\newnotation{\Prop}{\StaClass{Prop}}
\newnotation{\QF}{\StaClass{QF}}
\newnotation{\Ind}[1]{\text{IND}[#1]}

\notationnewgroup{Dynamic Complexity Classes}

\newnotation{\DynTC}{\DynClass{TC}}
\newnotation{\DynACz}{\DynClass{\ACz}}

\newnotation{\DynProp}{\DynClass{Prop}}
\newnotation{\DynQF}{\DynClass{QF}}
\newnotation{\DynFO}{\DynClass{FO}}
\newnotation{\DynFOar}{\DynClass{FO$(+,\times)$}}
\newnotation{\DynC}{\DynClass{$\class$}}[an arbitrary dynamic complexity class]

\notationnewgroup{Parameterised Complexity Classes}
\newnotation{\ParaP}{\ParaClass{\P}}
\newnotation{\FPT}{\StaClass{FPT}}
\newnotation{\ParaAC}{\ParaClass{\AC}}
\newnotation{\ParaTC}{\ParaClass{\TC}}
\newnotation{\ParaNC}{\ParaClass{\NC}}
\newnotation{\ParaACz}{\ParaClass{\ACz}}
\newnotation{\ParaACzu}{\ParaClass{\ACzu}}
\newnotation{\ParaSTD}{\ParaClass{ST-DynFO}}[parameterized structure and parameterized applications of formulas]
\newnotation{\ParaSD}{\ParaClass{S-DynFO}}[parameterized structure only]
\newnotation{\ParaTD}{\ParaClass{T-DynFO}}[parameterized  applications of formulas only]
\newnotation{\ParaSAC}{\ParaClass{S-\AC}}
\newnotation{\ParaTAC}{\ParaClass{T-\AC}}
\newnotation{\ParaSTAC}{\ParaClass{ST-\AC}}

\newnotation{\ParaST}{\ParaClass{ST-FO}}[parameterized structure, parameterized applications of the formula and arithmetic]
\newnotation{\ParaS}{\ParaClass{S-FO}}[parameterized structure and arithmetic]
\newnotation{\ParaT}{\ParaClass{T-FO}}[parameterized applications of the formula and arithmetic]

\newcommand{\paraProblemDescription}[4]{
  \setlength{\tabcolsep}{4pt}
\begin{flushleft}
  \begin{tabular}{r p{0.8\textwidth}}
      {\bf Problem:} & #1 \\
      {\bf Input:} & #2, %
      {\bf Parameter:}  #3 \\
      {\bf Question:} & #4 \\
  \end{tabular}
\end{flushleft}
}
\notationnewgroup{Non-Parameterised Problems}
\newnotation{\reach}{\problem{Reach}}
\newnotation{\nonPvCover}{\problem{VertexCover}}
\newnotation{\parity}{\problem{Parity}}

\notationnewgroup{Parameterised Problems}
\newnotation{\vCover}{\paraproblem{VertexCover}}
\newnotation{\planIndSet}{\paraproblem{PlanarIndependentSet}}
\newnotation{\cString}{\paraproblem{ClosestString}}
\newnotation{\cStringC}{\paraproblem{ClosestStringCandidate}}
\newnotation{\fVS}{\paraproblem{FeedbackVertexSet}}
\newnotation{\disFVS}{\paraproblem{DisjointFeedbackVertexSet}}
\newnotation{\disVCover}{\paraproblem{DisjointVertexCover}}
\newnotation{\mSPack}{\paraproblem{m-SetPacking}}
\newnotation{\SPack}{\paraproblem{SetPacking}}
\newnotation{\threeSPack}{\paraproblem{3-SetPacking}}
\newnotation{\knapsack}{\paraproblem{Knapsack}}
\newnotation{\longpath}{\paraproblem{LongestPath}}
\newnotation{\pointLC}{\paraproblem{$d$-PointLineCover}}

\notationnewgroup{General Math Macros}
\newnotation{\mtext}[1]{\textsc{#1}}[mathematical text (for function names etc.)]

\notationnewgroup{Calligraphic Letters}
\newnotationclass{callLetters}{\mathcal}
\ifcsname calA\endcsname\else\newnotation[callLetters]{\calA}{A}\fi
\ifcsname calB\endcsname\else\newnotation[callLetters]{\calB}{B}\fi
\ifcsname calC\endcsname\else\newnotation[callLetters]{\calC}{C}\fi
\ifcsname calD\endcsname\else\newnotation[callLetters]{\calD}{D}\fi
\ifcsname calE\endcsname\else\newnotation[callLetters]{\calE}{E}\fi
\ifcsname calF\endcsname\else\newnotation[callLetters]{\calF}{F}\fi
\ifcsname calG\endcsname\else\newnotation[callLetters]{\calG}{G}\fi
\ifcsname calH\endcsname\else\newnotation[callLetters]{\calH}{H}\fi
\ifcsname calI\endcsname\else\newnotation[callLetters]{\calI}{I}\fi
\newnotation[callLetters]{\calJ}{J}
\ifcsname calK\endcsname\else\newnotation[callLetters]{\calK}{K}\fi
\ifcsname calL\endcsname\else\newnotation[callLetters]{\calL}{L}\fi
\ifcsname calM\endcsname\else\newnotation[callLetters]{\calM}{M}\fi
\ifcsname calN\endcsname\else\newnotation[callLetters]{\calN}{N}\fi
\ifcsname calO\endcsname\else\newnotation[callLetters]{\calO}{O}\fi
\ifcsname calP\endcsname\else\newnotation[callLetters]{\calP}{P}\fi
\ifcsname calQ\endcsname\else\newnotation[callLetters]{\calQ}{Q}\fi
\ifcsname calR\endcsname\else\newnotation[callLetters]{\calR}{R}\fi
\ifcsname calS\endcsname\else\newnotation[callLetters]{\calS}{S}\fi
\ifcsname calT\endcsname\else\newnotation[callLetters]{\calT}{T}\fi
\ifcsname calU\endcsname\else\newnotation[callLetters]{\calU}{U}\fi
\ifcsname calV\endcsname\else\newnotation[callLetters]{\calV}{V}\fi
\ifcsname calW\endcsname\else\newnotation[callLetters]{\calW}{W}\fi
\ifcsname calX\endcsname\else\newnotation[callLetters]{\calX}{X}\fi
\ifcsname calY\endcsname\else\newnotation[callLetters]{\calY}{Y}\fi
\ifcsname calZ\endcsname\else\newnotation[callLetters]{\calZ}{Z}\fi

\notationnewgroup{Sets of Numbers}
\newnotationclass{setOfNumbers}{\mathbb} 
\newnotation[setOfNumbers]{\N}{N}[naturals]
\newnotation[setOfNumbers]{\Q}{Q}[rationals]
\newnotation[setOfNumbers]{\R}{R}[reals]

\notationnewgroup{Mathematical Notation}
\newnotation{\pmone}{\raisebox{.2ex}{$\scriptstyle\pm$}1}
\newnotation{\pone}{\raisebox{.2ex}{$\scriptstyle+$}1}
\newnotation{\mone}{\raisebox{.2ex}{$\scriptstyle-$}1}
\newnotation{\kmax}{k_\textup{max}}
\newnotation{\bigO}{\mathcal{O}}[big O]
\newnotation{\mneg}{\neg}[negation symbol]
\newnotation{\smallO}{o}[small o]
\newnotation{\hamD}{d_H}[Hamming distance]
\newnotation{\fvsSet}{\mathrm{FVS}}[Feedback Vertex Set as a set (not as problem)]
\newnotation{\vcSet}{\mathrm{VC}}[Vertex cover as a set (not as problem)]
\newnotation{\tpl}[1]{\bar{#1}}[$t$ is a tuple][\notationarg{1}{t}]
\newnotation{\restrict}[2]{#1[#2]}[$R$ is restricted to $A$][\notationarg{1}{R}\notationarg{2}{A}]
\newrelation{\df}{\smash{\stackrel{\scriptscriptstyle{\text{def}}}{=}}}[ definition]
\newrelation{\ubMulChanReach}{\frac{\log n}{\log \log n}}[An upper bound for multiple changes, after which reach can be maintained]

\makeatletter 

\notationnewgroup{Logic}
\newnotation{\eval}[3]{#1(#2/#3)}[The evaluation 
assigns values  $\overline{v}$ to variables $\overline{x}$ in
formula $\phi$][\notationarg{1}{\phi} \notationarg{2}{\overline{x}} \notationarg{3}{\overline{v}}]

\newnotation{\assignment}{\theta}[assignment]

\newnotation{\arity}{\text{Ar}}[arity]

\notationnewgroup{Schemata}
\newnotation{\schema}{\tau}
\notationnewgroup{Structures}
\newnotation{\struc}{\calS}
\newnotation{\struca}{\struc}
\newnotation{\strucb}{\calT}
\newnotation{\struct}[1]{\text{STRUC}[#1]}

\notationnewgroup{General Logical Notation}
\newnotation{\mthen}{\rightarrow}
\newnotation{\mand}{\wedge}
\newnotation{\mor}{\vee}
\newnotation{\munion}{\cup}
\newnotation{\mintersect}{\cap}
\newnotation{\mdisjunion}{\biguplus}

\newnotation{\arb}{\star}%
\newnotation{\generic}{\textsc{generic}}
\newnotation{\quant}{\mathbb{Q}}
\newnotation{\cquant}{\overline{\mathbb{Q}}}

\notationnewgroup{Databases}
\newnotation{\db}{\calD}
\newnotation{\inp}{\calI}
\newnotation{\aux}{\calA}
\newnotation{\para}{\Pi}
\newnotation{\builtin}{\calB}
\newnotation{\domain}{D}
\newnotation{\actDomain}{D_\text{act}}
\newnotation{\emptyDB}{\db_\emptyset}

\newnotation{\rel}[1]{#1}[relation layout]

\newnotation{\query}{Q}[query]
\newnotation{\cq}{\calC}[conjunctive query]
\newnotation{\Qreach}{\query_{\text{Reach}}}[The Reachability Query]
\newnotation{\ans}[2]{\mtext{ans}(#1, #2)}[the answer of a query~$q$ on database~$db$][\notationarg{1}{q}\notationarg{2}{db}]
\newnotation{\qRel}{\mtext{ans}}[name of the query relation]

\newnotation{\adom}{\text{adom}}[active domain]
\newnotation{\dom}{\text{dom}}[domain]

\newnotation{\BIT}{\text{BIT}}[the BIT predicate]
\newnotation{\leqadom}{\leq_{\adom}}[linear order on active domain]
\newnotation{\BITadom}{\text{BIT}_{\adom}}

\notationnewgroup{Dynamic Complexity}
\newnotation{\updates}{\Delta}
\newnotation{\abstrDel}{\updates_{Del}}
\newnotation{\abstrIns}{\updates_{Ins}}
\newnotation{\abstrUpd}{\updates}
\newnotation{\init}{\mtext{Init}}
\newnotation{\ins}{\mtext{ins}}
\newnotation{\del}{\mtext{del}}
\newnotation{\set}{\mtext{set}}
\newnotation{\rep}{\mtext{rep}}
\newnotation{\insertdescr}[2]{\textbf{Insertion of }#2\textbf{ into }#1.}
\newnotation{\deletedescr}[2]{\textbf{Deletion of }#2\textbf{ from } #1.}

\newnotation{\state}{\struc}[state]

\notationnewgroup{Dynamic Complexity Schemata}
\newnotation{\inpSchema}{\schema_{\textup{in}}}[input schema]
\newnotation{\auxSchema}{\schema_{\textup{aux}}}[auxiliary schema]
\newnotation{\aritSchema}{\schema_{\textup{arith}}}[arithmetic schema]
\newnotation{\advSchema}{\schema_{\textup{adv}}}
\newnotation{\eqSchema}{\schema_{=}}
\newnotation{\builtinSchema}{\schema_{\textup{bi}}}
\newnotation{\modSchema}{\schema_{\textup{mod}}}

\notationnewgroup{Initialization}
\newnotation{\auxInit}{\init_{\textup{aux}}}
\newnotation{\builtinInit}{\init_{\textup{bi}}}

\newnotation{\FOprog}{\calF}

\notationnewgroup{Program Names}
\newnotation{\upProg}{P}
\newnotation{\prog}{\calP}
\newnotation{\progb}{Q}

\notationnewgroup{Program Updates}
\newnotation{\updateDB}[2]{#1(#2)}[Application of update $\alpha$ in database $db$][\notationarg{1}{\alpha}\notationarg{2}{db}]
\newnotation{\updateState}[3]{#1_{#2}(#3)}[Program $P$ applies
update $\alpha$ to state $S$][\notationarg{1}{P} \notationarg{2}{\alpha}
\notationarg{3}{S}]
\newnotation{\updateRelation}[4]{\restrict{{#1}_{#2}(#3)}{#4}}

\newnotation{\transition}[3]{{#1} \xrightarrow{#2}{#3}}

\notationnewgroup{Update Formulas}

\newnotation{\uf}[4]{
  \@ifmtarg{#4}{
    \phi^{#1}_{#2}(#3)}{
    \phi^{#1}_{#2}(#3; #4)}
}[1: updated auxiliary relation, 2: update, 3: inserted/deleted tuple, 
4 updated tuple in auxiliary relation]

\newnotation{\huf}[4]{
\@ifmtarg{#4}{\widehat{\phi}^{#1}_{#2}(#3)}{\widehat{\phi}^{#1}_{#2}(#3; #4)}
}[update formula with hat]

\newnotation{\ufb}[4]{
\@ifmtarg{#4}{%
\psi^{#1}_{#2}(#3)
}{
\psi^{#1}_{#2}(#3; #4)
}}

\newnotation{\ufbwa}[2]{\psi^{#1}_{#2}}

\newnotation{\ufwa}[2]{\phi^{#1}_{#2}}[update formula withour arguments. 1: updated auxiliary relation, 2: update]

  \newnotation{\ufsubstitute}[5]{
    \@ifmtarg{#5}{
      \phi^{#2}_{#3}[#1](#4)
    }{
      \phi^{#2}_{#3}[#1](#4; #5)
    }
  }[Modified update formula. 1: substitute, 2: updated auxiliary relation, 
  3: update, 4: inserted/deleted tuple, 5 updated tuple in auxiliary relation]
\newnotation{\ufsubstitutewa}[3]{ \phi^{#2}_{#3}[#1]}
\newnotation{\substitutewa}[2]{#1[#2]}[Modified update formulas
 1: substitute, 2: updated auxiliary relation, 3: update, 4: inserted/deleted tuple, 5 updated tuple in auxiliary relation]

\notationnewgroup{Update Terms}
\newnotation{\ut}[4]{
  \@ifmtarg{#4}{t^{#1}_{#2}(#3) }{t^{#1}_{#2}(#3; #4)}
}[update term. 1: updated auxiliary function, 2: update, 3: inserted/deleted tuple, 4 updated tuple in auxiliary relation]

\newnotation{\utwa}[2]{t^{#1}_{#2}}
\newnotation{\ite}[3]{
  \@ifmtarg{#1}{
    \mtext{ITE}
   }{
    \mtext{ITE}\text{(#1,#2,#3)}  
  }
}[if-then-else update term]

\newnotation{\mf}[3]{\@ifmtarg{#3}{\mu_{#1}(#2)}{\mu_{#1}(#2; #3)}}[modification formula. 1: modified input relation,  2: parameters, 3 modified tuple in input relation]

\newnotation{\mfos}[4]{
  \@ifmtarg{#4}{
    {#1}_{#2}(#3)}{
    {#1}_{#2}(#3; #4)}
}[modification formula with other symbol. 1: other symbol, 2: modified input relation, 3: parameters, 4 modified tuple in input relation]

\newnotation{\mfwa}[1]{\mu_{#1}}[modification formula without arguments
1: updated auxiliary relation]

  \newcommand{\changeRule}[3]{
          \textbf{on change}\ #1\ \textbf{update}\ #2\ \textbf{as}\ #3%
  }

\makeatother

\ifcomments
\newcommand{\commentbox}[1]{\noindent\framebox{\parbox{0.98\linewidth}{#1}}}

\newcommand{\acomment}[2]{\ \\ \fbox{\parbox{0.98\linewidth}{{\sc #1}: #2}}}
\newcommand{\mcomment}[2]{{\color{blue}(#1)}\footnote{#1: #2}} %
\else
\newcommand{\commentbox}[1]{}
\newcommand{\mcomment}[2]{}
\newcommand{\acomment}[2]{}
\fi

\ifchanges

\setul{}{0.2mm}
\setstcolor{red}

\else

\fi

 \newcommand{\tzm}[1]{\mcomment{TZ}{#1}}
 \newcommand{\tsm}[1]{\mcomment{TS}{#1}}
 \newcommand{\nilsm}[1]{\mcomment{NV}{#1}}

\newcommand*{\fancyrefexlabelprefix}{ex}

\newcommand*{\frefexname}{\text{example}}
\newcommand*{\Frefexname}{\text{Example}}

\frefformat{vario}{\fancyrefexlabelprefix}{\frefexname\fancyrefdefaultspacing#1#3}
\Frefformat{vario}{\fancyrefexlabelprefix}{\Frefexname\fancyrefdefaultspacing#1#3}
\frefformat{plain}{\fancyrefexlabelprefix}{\frefexname\fancyrefdefaultspacing#1}
\Frefformat{plain}{\fancyrefexlabelprefix}{\Frefexname\fancyrefdefaultspacing#1}

\newcommand*{\fancyrefclaimlabelprefix}{claim}

\newcommand*{\frefclaimname}{\text{claim}}
\newcommand*{\Frefclaimname}{\text{Claim}}

\frefformat{vario}{\fancyrefclaimlabelprefix}{\frefclaimname\fancyrefdefaultspacing#1#3}
\Frefformat{vario}{\fancyrefclaimlabelprefix}{\Frefclaimname\fancyrefdefaultspacing#1#3}
\frefformat{plain}{\fancyrefclaimlabelprefix}{\frefclaimname\fancyrefdefaultspacing#1}
\Frefformat{plain}{\fancyrefclaimlabelprefix}{\Frefclaimname\fancyrefdefaultspacing#1}

\newcommand*{\fancyrefdeflabelprefix}{def}

\newcommand*{\frefdefname}{\text{definition}}
\newcommand*{\Frefdefname}{\text{Definition}}

\frefformat{vario}{\fancyrefdeflabelprefix}{\frefdefname\fancyrefdefaultspacing#1#3}
\Frefformat{vario}{\fancyrefdeflabelprefix}{\Frefdefname\fancyrefdefaultspacing#1#3}
\frefformat{plain}{\fancyrefdeflabelprefix}{\frefdefname\fancyrefdefaultspacing#1}
\Frefformat{plain}{\fancyrefdeflabelprefix}{\Frefdefname\fancyrefdefaultspacing#1}

\newcommand*{\fancyrefthmlabelprefix}{thm}

\newcommand*{\frefthmname}{\text{theorem}}
\newcommand*{\Frefthmname}{\text{Theorem}}

\frefformat{vario}{\fancyrefthmlabelprefix}{\frefthmname\fancyrefdefaultspacing#1#3}
\Frefformat{vario}{\fancyrefthmlabelprefix}{\Frefthmname\fancyrefdefaultspacing#1#3}
\frefformat{plain}{\fancyrefthmlabelprefix}{\frefthmname\fancyrefdefaultspacing#1}
\Frefformat{plain}{\fancyrefthmlabelprefix}{\Frefthmname\fancyrefdefaultspacing#1}

\newcommand*{\fancyrefremlabelprefix}{rem}

\newcommand*{\frefremname}{\text{remark}}
\newcommand*{\Frefremname}{\text{Remark}}

\frefformat{vario}{\fancyrefremlabelprefix}{\frefremname\fancyrefdefaultspacing#1#3}
\Frefformat{vario}{\fancyrefremlabelprefix}{\Frefremname\fancyrefdefaultspacing#1#3}
\frefformat{plain}{\fancyrefremlabelprefix}{\frefremname\fancyrefdefaultspacing#1}
\Frefformat{plain}{\fancyrefremlabelprefix}{\Frefremname\fancyrefdefaultspacing#1}

\newcommand*{\fancyreflemlabelprefix}{lem}

\newcommand*{\freflemname}{\text{lemma}}
\newcommand*{\Freflemname}{\text{Lemma}}

\frefformat{vario}{\fancyreflemlabelprefix}{\freflemname\fancyrefdefaultspacing#1#3}
\Frefformat{vario}{\fancyreflemlabelprefix}{\Freflemname\fancyrefdefaultspacing#1#3}
\frefformat{plain}{\fancyreflemlabelprefix}{\freflemname\fancyrefdefaultspacing#1}
\Frefformat{plain}{\fancyreflemlabelprefix}{\Freflemname\fancyrefdefaultspacing#1}

\newcommand*{\fancyrefsubseclabelprefix}{subsec}

\newcommand*{\frefsubsecname}{\text{subsection}}
\newcommand*{\Frefsubsecname}{\text{Subsection}}

\frefformat{vario}{\fancyrefdeflabelprefix}{\frefsubsecname\fancyrefdefaultspacing#1#3}
\Frefformat{vario}{\fancyrefsubseclabelprefix}{\Frefsubsecname\fancyrefdefaultspacing#1#3}
\frefformat{plain}{\fancyrefsubseclabelprefix}{\frefsubsecname\fancyrefdefaultspacing#1}
\Frefformat{plain}{\fancyrefsubseclabelprefix}{\Frefsubsecname\fancyrefdefaultspacing#1}

\newcommand*{\fancyrefcorlabelprefix}{cor}

\newcommand*{\frefcorname}{\text{corollary}}
\newcommand*{\Frefcorname}{\text{Corollary}}

\frefformat{vario}{\fancyrefcorlabelprefix}{\frefcorname\fancyrefdefaultspacing#1#3}
\Frefformat{vario}{\fancyrefcorlabelprefix}{\Frefcorname\fancyrefdefaultspacing#1#3}
\frefformat{plain}{\fancyrefcorlabelprefix}{\frefcorname\fancyrefdefaultspacing#1}
\Frefformat{plain}{\fancyrefcorlabelprefix}{\Frefcorname\fancyrefdefaultspacing#1}

\newcommand*{\fancyrefalgolabelprefix}{algo}

\newcommand*{\frefalgoname}{\text{algorithm}}
\newcommand*{\Frefalgoname}{\text{Algorithm}}

\frefformat{vario}{\fancyrefalgolabelprefix}{\frefalgoname\fancyrefdefaultspacing#1#3}
\Frefformat{vario}{\fancyrefalgolabelprefix}{\Frefalgoname\fancyrefdefaultspacing#1#3}
\frefformat{plain}{\fancyrefalgolabelprefix}{\frefalgoname\fancyrefdefaultspacing#1}
\Frefformat{plain}{\fancyrefalgolabelprefix}{\Frefalgoname\fancyrefdefaultspacing#1}

\newcommand*{\fancyrefproplabelprefix}{prop}

\newcommand*{\frefpropname}{\text{proposition}}
\newcommand*{\Frefpropname}{\text{Proposition}}

\frefformat{vario}{\fancyrefproplabelprefix}{\frefpropname\fancyrefdefaultspacing#1#3}
\Frefformat{vario}{\fancyrefproplabelprefix}{\Frefpropname\fancyrefdefaultspacing#1#3}
\frefformat{plain}{\fancyrefproplabelprefix}{\frefpropname\fancyrefdefaultspacing#1}
\Frefformat{plain}{\fancyrefproplabelprefix}{\Frefpropname\fancyrefdefaultspacing#1}

\tikzstyle{mnode}=[
  circle,
  fill=\mnodefillcolor, 
  draw=\mnodedrawcolor,
  minimum size=6pt, 
  inner sep=0pt
]

\tikzstyle{mnodeinvisible}=[
  minimum size=6pt, 
  inner sep=0pt
]

\tikzstyle{invisible}=[
  minimum size=0pt, 
  inner sep=0pt
]

\tikzstyle{invisiblel}=[
  minimum size=10pt, 
  inner sep=0pt
]

\tikzstyle{invisibleEdge}=[
  transparent
]

\tikzstyle{nameNode}=[
  font=\scriptsize
]

\tikzstyle{namingNode}=[
  font=\normalsize
]

\tikzstyle{mEdge}=[
  -latex', %
  thick, 
  shorten >=3pt, 
  shorten <=3pt,
  draw=black!80,
]

\tikzstyle{dDashedEdge}=[
  -latex', %
  thick, 
  shorten >=3pt, 
  shorten <=3pt,
  draw=black!80,
  dashed
]

\tikzstyle{dEdge}=[
  -latex', %
  thick, 
  shorten >=1pt, 
  shorten <=1pt,
  draw=black!80,
]

\tikzstyle{dhEdge}=[
  -latex', %
  thick, 
  shorten >=3pt, 
  shorten <=3pt,
  draw=black!80,
]
\tikzstyle{uEdge}=[
  thick, 
  shorten >=3pt, 
  shorten <=3pt,
  draw=black!80,
]
\tikzstyle{uhEdge}=[
  thick, 
  shorten >=3pt, 
  shorten <=3pt,
  draw=black!80,
]

\tikzstyle{cEdge}=[
  ultra thick, 
  shorten >=3pt, 
  shorten <=3pt,
  draw=black!80,
]

\tikzstyle{dotsEdge}=[
  very thick, 
  loosely dotted, 
  shorten >=7pt, 
  shorten <=7pt
]

\tikzstyle{class rectangle}=[
  draw=black,
  inner sep=0.2cm,
  rounded corners=5pt,
  thick
]

\tikzstyle{mline}=[
  draw=black,
  inner sep=0.2cm,
  rounded corners=5pt,
  thick
]

\tikzstyle{mainclass rectangle}=[
  draw=blue,
  inner sep=0.2cm,
  rounded corners=5pt,
  very thick
]

\newcommand{\mnodedrawcolor}{black!80}
\newcommand{\mnodefillcolor}{black!40}

\pgfdeclarelayer{background}
\pgfdeclarelayer{substructure}
\pgfdeclarelayer{edges}
\pgfdeclarelayer{foreground}
\pgfsetlayers{background,substructure,edges,main,foreground}

\tikzstyle{background rectangle}=[
  fill=black!10,
  draw=black!20,
  line width = 5pt,
  inner sep=0.4cm,
  outer sep=0.4cm,
  rounded corners=5pt
] 
\begin{document}
  \maketitle
 \begin{abstract}
Dynamic Complexity studies the maintainability of queries with logical
formulas in a setting where the underlying structure or database
changes over time. Most often, these formulas are  from
first-order logic, giving rise to the dynamic complexity class
\DynFO. This paper investigates extensions of \DynFO in the spirit of
parameterised algorithms. In this setting structures come with a
parameter~$k$ and the extensions allow additional ``space''  of size
$f(k)$ (in the form of an additional structure of this size) or
additional time  $f(k)$ (in the form of iterations of formulas) or
both. The resulting classes are compared with their non-dynamic
counterparts and other classes. The main part of the paper explores
the applicability of methods for parameterised algorithms to this
setting through case studies for various well-known parameterised
problems.  
 \end{abstract}

 \section{Introduction}\label{section:introduction}
Parameterised complexity studies aspects of problems that make them computationally hard. %
The main interest has been in the class \FPT %
which subsumes all problems that can be solved in time $f(k) \textsf{poly}(|x|)$ for an input $x$ with a \emph{parameter} $k \in \N$ and a computable function~$f$.
In recent work, much smaller parameterised classes have been studied,
derived from classical classes in a uniform way by replacing the
requirement of a polynomial bound of e.g.~the circuit size (time, space, \dots, respectively) by a bound of the
form $f(k) \textsf{poly}(|x|)$.   In this fashion classical circuit
classes $\AC^i$ and $\NC^i$ naturally translate to parameterised
classes $\ParaAC^i$ and $\ParaNC^i$.  The lowest of these classes,
$\ParaAC^0$ corresponds to the class $\AC^0$ of problems computable by
uniform families of constant-depth, polynomial size circuits with $\wedge$-,
$\vee$- and $\neg$-gates of unbounded fan-in \cite{ElberfeldST15,BannachST15}. 

This paper adds the aspect of changing inputs and dynamic
maintenance of results  to the exploration of the landscape between  $\ParaAC^0$ and \FPT.

The study of low-level complexity classes under dynamic aspects was
started in \cite{PatnaikI97, DongST95} in the context of dynamically
maintaining the result of database queries. Similarly, as for
\emph{dynamic algorithms}, in this setting a dynamic program can make
use of auxiliary relations that can store knowledge about the current
input data (database). After a small change of the database (most
often:  insertion or deletion of a tuple), the program needs to
compute the query result for the modified database in very short
parallel time. To capture the problems/queries, for which this is
possible, Patnaik and Immerman introduced the class \DynFO
\cite{PatnaikI97}. Here, ``FO'' stands for first-order logic, which is
equivalent to $\AC^0$,  in the presence of arithmetic \cite{BarringtonIS90, ImmermanDC}.

In this paper, we study dynamic programs that have additional
resources in a ``parameterised sense''. We explore two
such resources, which can be  described as \emph{parameterised space} and \emph{parameterised time}, respectively.  For ease of exposition, we discuss these two resources in the context of $\AC^0$ first. %

One way to strengthen $\AC^0$ circuit families  is to allow circuits of \emph{size} $f(k) \textsf{poly}(|x|)$.  %
We denote the class thus obtained as $\ParaSAC^0$ (even though it corresponds to the class $\ParaAC^0$). A second dimension is to let the \emph{depth} of circuits  depend on the parameter. As the depth of circuits corresponds to the (parallel) time the circuits need for a computation, we denote the class of problems captured by such circuits by $\ParaTAC^0$. Of course, both dimensions can also be combined, yielding the parameterised class $\ParaSTAC^0$. 

Surprisingly, several parameterised versions of \NP-complete problems can even
be solved in $\ParaSAC^0$. Examples are 
the vertex cover problem and the hitting set problem parameterised by
the size of the vertex cover and the hitting set, respectively
\cite{BannachT18b}. However, classical circuit lower bounds unconditionally imply that
this is not possible for all \FPT-problems. For instance, in~\cite{BannachST15} it was observed that the
existence of simple paths of length $k$ (the parameter) cannot be
tested in $\ParaSAC^0$. Likewise, the feedback vertex set problem with
the size of the feedback vertex set as parameter cannot be solved in
$\ParaSTAC^0$. %

When translated from circuits to logical formulas, depth roughly translates into iteration of
formulas~\cite[Theorem 5.22]{ImmermanDC},
whereas size translates into the size of an additional structure by which the database is extended before formulas are evaluated. Slightly more formally, $\ParaTAC^0$ corresponds to the class $\ParaT$ consisting of problems that can be defined by iterating a formula $f(k)$ many times.  The class $\ParaSAC^0$ corresponds to the class $\ParaS$ where formulas are evaluated on structures $\db$ extended by an \emph{advice structure} whose size depends on the parameter only. In the class $\ParaST$ both dimensions are combined.
The parameterised \emph{dynamic} classes that we study in this paper are obtained from \DynFO just like the above classes are obtained from \FO: \ParaSD, \ParaTD and \ParaSTD extend \DynFO by an additional structure of parameterised size,  $f(k)$ iterations of formulas, or both, respectively. 

  As our first main contribution, we introduce a uniform framework for small dynamic,
  parameterised complexity classes (\Fref{section:setting}) based on
  advice structures (corresponding to additional space) or iterations of formulas (corresponding to additional time) and investigate how the
  resulting classes relate to each other and to other non-dynamic (and
  even non-parameterised) complexity classes (\Fref{section:relations}).

 As our second main contribution, we explore how methods for parameterised algorithms can
  be applied in this framework through case studies for various
  parameterised problems (\Fref{section:methods}). %
Due to space limitations, many proofs are delegated to the appendix.

\subparagraph*{Related work}
There is a rich literature on parameterised   dynamic algorithms, e.g.~\cite{HartungN13, downey2014dynamic, MansM17, BockenhauerBRR18reopt, AlmanMW17}. %
Closer to our work
is the investigation of (static) parameterised small (parallel)
complexity classes that was initiated
$20$ years ago in \cite{CesatiI98}. %
Later, in~\cite{ElberfeldST15}, parameterised
versions of space and circuit classes were defined
and several known parameterised problems were
shown to be complete for these classes.
Also in \cite{BannachST15} it was shown,
by applying the colour-coding technique, that
several parameterised problems belong in
\ParaACz. Furthermore Chen and Flum~\cite{ChenF16}
presented some unconditional proofs showing that
some parameterised problems do not belong
in \ParaACz. 

The descriptive complexity of parameterised classes
has also been investigated in the past. For example
Flum and Grohe~\cite{FlumG03} and
Bannach and Tantau~\cite{BannachT19}
presented syntactic descriptions of parameterised
complexity classes using logical formulas. Additionally
Chen, Flum and Huang~\cite{ChenFH17} showed that the $k$-slices
of several problems can be defined using
\FO-formulas of quantifier rank independent of $k$
and explored the connection between the quantifier rank of \FO-sentences and the depth of \ACz-circuits.

 \section{Preliminaries}\label{section:preliminaries}
By $[n]$ we denote the set $\{1, \ldots , n\}$.  %
We assume familiarity with first-order logic \FO   and refer to \cite{Libkin04} for  basics of finite model theory. A \emph{(relational) schema} $\schema$ consists of a set of relation symbols with a corresponding arity. %
A \emph{structure} $\db$ over schema $\schema$ with domain $\domain$ has, for every relation symbol $R \in \schema$, a relation over $\domain$ with the same arity as $R$. Throughout this work domains are finite.  
A \emph{$k$-ary query} $\query$ on $\schema$-structures is a mapping that assigns a subset of $\domain^k$ to every $\schema$-structure over domain~$\domain$ and commutes with isomorphisms. Each first-order formula $\varphi(\tpl x)$ over schema $\schema$ defines a query $\query$ whose result on a $\schema$-structure $\db$ is $\{\tpl a \mid \db \models \varphi(\tpl a)\}$. %
Queries of arity 0 are also called \emph{Boolean queries} or \emph{problems}. %

We mainly consider first-order formulas that have access to arithmetic, that is to a linear order $<$ on the domain as well as suitable, compatible addition $+$ and multiplication $\times$. We require  that the result of the formulas is invariant\footnote{In our scenario it is not relevant that invariance is undecidable for first-order formulas.} under the choice of the linear order~$<$. This logic is referred to as \emph{order-invariant first-order logic with arithmetic} and denoted by \FOar. In linearly ordered domains, we often identify domain elements with natural numbers, the smallest element representing 1. %

\subparagraph*{Dynamic Complexity}
We work in the dynamic complexity framework as introduced by Patnaik and Immerman \cite{PatnaikI97}, and refer to \cite{SchwentickZ16} for details. In a nutshell, dynamic programs answer a query for an input structure that is subjected to a sequence of changes. To this end they maintain an auxiliary structure using logical formulas. 

 By $\Delta_\schema$ we denote the set of \emph{single-tuple change operations} for a schema $\schema$, which consists of the insertion operations $\ins_R$ and the deletion operations $\del_R$ for each relation $R \in \schema$.  For example, $\ins_E(a,b)$ could add edge $(a,b)$ to a graph.
 A \emph{dynamic query} $(\query, \Delta)$ consists of a query~$\query$ over some input schema $\inpSchema$ and a set $\Delta \subseteq \Delta_\inpSchema$. Later on we will sometimes consider slightly more general change operations.
 
A \emph{dynamic program} $\prog$ for a dynamic query $(\query, \Delta)$ continuously answers  $\query$  on an \emph{input structure}~$\inp$ over some \emph{input schema} $\inpSchema$ under changes of the input structure from $\Delta$.   The domain $\domain$ of $\inp$ is fixed and in particular changes cannot introduce new elements.\footnote{We note that this is not a severe restriction, see e.g.~\cite[Theorem 17]{DattaKMSZ18}.} The program $\prog$  maintains an \emph{auxiliary structure}~$\aux$ over some \emph{auxiliary schema} $\auxSchema$ with the same domain as $\inp$. We call $(\inp, \aux)$ a \emph{state} of $\prog$ and consider it as one relational structure.
The auxiliary structure includes one particular \emph{query relation} $\qRel$ that is supposed to contain the answer of $\query$ over $\inp$.
For each   auxiliary relation $S \in \auxSchema$ and each change operation $\delta \in \Delta$, \prog has an update rule that specifies how $S$ is updated after a change. It is of the form \changeRule{$\delta(\tpl p)$}{$S(\tpl x)$}{$\uf{S}{\delta}{\tpl p}{\tpl x}$} where the \emph{update formula}  $\uf{S}{\delta}{\tpl p}{\tpl x}$ is a formula over~$\inpSchema \cup \auxSchema$. For example, if the tuple $\tpl a$ is inserted into an input relation $R$, each auxiliary relation $S$ is replaced by the relation $\{ \tpl b \mid (\inp, \aux) \models \uf{S}{\ins_R}{\tpl a}{\tpl b} \}$.
By $\alpha(\inp)$  we denote the input structure that results from $\inp$ by applying a sequence $\alpha$ of changes, and by $\prog_\alpha(\inp, \aux)$ the state $(\alpha(\inp), \aux')$ of $\prog$ that results from $(\inp, \aux)$ after processing~$\alpha$.
The dynamic program $\prog$ \emph{maintains} $(\query, \Delta)$ if the relation $\qRel$ in $\prog_\alpha(\inp_0, \aux_0)$ equals the query result $\query(\alpha(\inp_0))$, for each sequence $\alpha$ of changes over $\Delta$, each initial input structure $\inp_0$ with arbitrary (finite) domain and empty relations, and the auxiliary structure $\aux_0$ with empty relations.

The class \DynFO is the set of dynamic queries that can be maintained by a dynamic program with first-order update formulas. The class \DynFOar is defined analogously via \FOar update formulas.
We note that in the case of \DynFOar, we consider the arithmetic relations to be part of the input structure \inp, but they can not be modified. Technically, an additional schema \aritSchema contains the arithmetic predicates and the update formulas are over $\inpSchema \cup \auxSchema \cup \aritSchema$.
Note that \aritSchema cannot be used for defining a query. %

\subparagraph*{Parameterised Complexity}
A \emph{parameterised query} is a pair $(\query, \kappa)$, where $\query$ is a query over some schema $\tau$ and $\kappa$ is a function, called the \emph{parameterisation}, that assigns a parameter from~$\N$ to every $\tau$-structure. The well-known parameterised complexity class \FPT contains all Boolean parameterised queries $(\query, \kappa)$ having an algorithm that decides for each $\tau$-structure $\db$ whether $\db \in \query$ in time $f(\kappa(\db)) |\db|^c$, for some constant $c$ and computable function $f \colon \N \to \N$~\cite{DowneyF95}. Like \cite{BannachT18},  we demand that $\kappa$ is first-order definable, which is always the case if the parameter is explicitly given in the input.\tsm{Discussion of parameter treewidth in commments.}

\begin{example}\label{ex:vc-tree-fpt}
  \vCover is a well-studied parameterised query. Formally it is the set $\query$ of pairs $(G,k)$, where $G$ is an undirected graph that has a vertex cover of size $k$, together with the parameterisation $\kappa \colon (G,k) \mapsto k$. In  more accessible notation:
\paraProblemDescription{\vCover}{An undirected graph $G=(V,E)$ and $k \in \N$}{$k$}{Is there a set $S\subseteq V$ such that $|S|=k$ and $u \in S$ or $v \in S$ for every $(u,v) \in E$?}

  The search-tree based algorithm for \vCover is a classical
  parameterised algorithm. It is based on
  the simple observation that, for each edge $(u,v)$ of a graph, each
  vertex cover needs to contain $u$ or $v$ (or both). 
  On input $(G,k)$
  the algorithm recursively constructs the search tree as follows, starting from the
  root of an otherwise empty tree. If $E$ is empty
  it accepts, otherwise it rejects if
  $k=0$.     If $k>0$  it chooses some edge $(u,v)\in E$, labels the
  current node with $(u,v)$,
  and constructs two new tree nodes below the current node. It then
  continues recursively, from both children starting from the instance
  $(G-u,k-1)$ in the first child, and from  $(G-v,k-1)$ in the second
  child. The algorithm accepts if any of its branches accepts.
  Since the inner nodes of the tree have two children and its depth is
  bounded by $k$, it can have at most $2^{k+1}-1$ tree nodes. The
  overall running time can be bounded by  $\bigO(2^kn^2)$. Thus
  $\vCover\in\FPT$.
  \qed
\end{example}

\tsm{I tentatively removed three paragraphs on \para-classes here. Some of that might come back, maybe somewhat later. }

\section{A Framework for Parameterised, Dynamic Complexity}\label{section:setting}
We first present a uniform point of view on parameterised first-order logic. %
As explained in the introduction, formulas can be parameterised with respect to (at least) two dimensions:   additional time by iterating formulas with the number of iterations depending on the parameter; additional space by advice structures whose size depends on the parameter. 

A \emph{first-order program} $\FOprog$ over schema $\schema$ is a tuple $(\Psi, \varphi)$ where $\Psi$ is a set of \FOar-formulas over schema $\schema \uplus \schema_\Psi$ and $\varphi \in \Psi$ is supposed to compute the final result of the program. Here, $\schema_\Psi$ is a schema that contains a fresh relation symbol $R_\psi$ for each formula $\psi \in \Psi$ of the same arity as $\psi$. 
The semantics of $\FOprog$ on a $\schema$-structure $\db$ is based on  inductively  defined $\schema_\Psi$-structures $\db^{(\ell)}_\Psi$.  Initially, in~$\db^{(0)}_\Psi$, all relations $R^{(0)}_\psi$ are empty. 
The \emph{$\ell$-step result} $\db^{(\ell)}_\Psi$ of $\FOprog$, for $\ell>0$, is defined via $R^{\ell}_\psi \df \{ \tpl a \mid (\db, \db^{(\ell-1)}_\Psi) \models \psi(\tpl a)\}$. Finally, the \emph{result} $\FOprog(\db)$ is $R^{(\ell)}_\varphi$  if $\db^{(\ell-1)}_\Psi=\db^{(\ell)}_\Psi$, for some $\ell$. In this case, we say that the program reaches a fixed point after $\ell$ steps. Otherwise, $\FOprog(\db)$ is the empty set. 

We now define how first-order programs can use advice. An \emph{$\advSchema$-advice} $\pi$ is a computable mapping from $\N$ to $\advSchema$-structures for some fixed advice schema $\schema_{\text{adv}}$. Suppose that $\FOprog$ is a first-order program over schema $\schema \uplus \advSchema$. The result of $\FOprog$ for a $\schema$-structure $\db$ with advice $\pi$ and parameter $k \in \N$ is simply the result of $\FOprog$ on the structure $\db \uplus \pi(k)$.

For two computable functions $f, g: \N \rightarrow \R$ and a parameterised query $(\query, \kappa)$ over a schema $\schema$, an \emph{$(f, g)$-parameterised first-order program for $(\query, \kappa)$} is a tuple $(\FOprog, \pi)$ where $\FOprog$ is a first-order program over schema $\schema \uplus \advSchema$ and $\pi$ is an $\advSchema$-advice such that 
\begin{enumerate}
\item the result of $\FOprog$ with advice $\pi$ is $\query(\db)$, for all $\schema$-structures $\db$;
 \item $|\pi(\kappa(\db))| \leq f(\kappa(\db))$ for all $\schema$-structures $\db$; and 
 \item  $\FOprog$ always reaches a fixed point and does so after at most $g(\kappa(\db))$ steps.
\end{enumerate}

For computable functions $f$ and $g$ let $\ParaST(f,g)$ be the class of parameterised queries definable by an $(f, g)$-parameterised first-order program. 
We note that these programs use \FOar formulas, and thus have access to arithmetic\footnote{In particular, ``$+|\db|$'' induces a correspondence between  $\db$ and $\pi(k)$.}
over the domain of $\db \uplus \pi(k)$. We do not make this explicit in our naming scheme. We use the following abbreviations:
\begin{itemize}
\item $\ParaST \df \bigcup_{f,g} \ParaST(f,g)$,
\item $\ParaS \df \bigcup_{f} \ParaST(f, 1)$, 
\item $\ParaT \df \bigcup_{g} \ParaST(0, g)$. 
\end{itemize}

The class \ParaS is in fact the same as \ParaACz, and \ParaST corresponds to the class \ParaACzu in \cite{BannachST15}. To the best of our knowledge,  $\ParaT$  has not been studied in the context of first-order logic before. %

\begin{example}
\label{ex:vc-t}
We sketch a first-order program $\FOprog = (\Psi, \varphi)$ that witnesses $\vCover \in \ParaT$. Recall the search-tree based parameterised algorithm for \vCover from Example~\ref{ex:vc-tree-fpt}. 
Intuitively, the formulas $\psi \in \Psi$ are used to traverse the search tree %
in a depth-first manner. At any moment, the auxiliary relations contain information about the path from the root to the current node. In particular, the \emph{candidate set} of the current node, i.e., the set of vertices selected along its path is available.
Each application of these formulas simulates one elementary step of the search: either a new child is added to the current path, or, if the current node has maximal depth or if all possible children were already added, the current node is discarded and a backtrack step to its parent is performed.
If the candidate set is a vertex cover, the search ends. Since each edge of the search tree needs to be traversed at most twice, $2^{k+2}$ iterative steps suffice.  
More detail is given in the appendix.
\qed %
\end{example}

\begin{toappendix}
  {\bf Continuation of Example~\ref{lem:small}.}

    We give some more details. The program $\FOprog$ uses a ternary relation $P$ and a binary relation $C$ with the intention that, at each point of the computation, $P$ contains a tuple $(i, u, v)$ if the edge $(u, v)$ is used at level $i$ on the current path of the tree  and $C$ contains a tuple $(i, w)$ if the vertex $w$ is chosen for the candidate set at level $i$ on this path. The relations are modified  as follows. The formula $\varphi$ defines the empty set as long as $C$ does \emph{not} encode a valid vertex cover. As soon as $C$ \emph{does} encode a valid vertex cover, $\varphi$ defines this set. As soon as $\varphi$ defines a non-empty set all formulas reproduce their previous result and thus a fixed-point is reached.  Before that happens, we distinguish the following cases:
\begin{itemize}
 \item If $C$ contains $\ell < k$ %
tuples, the lexicographically smallest edge $(u,v)$ with no endpoint in $C$ is selected. The tuple $(\ell+1, u, v)$ is inserted into $P$ and the tuple $(\ell + 1, u)$ is inserted into $C$.
 \item If $C$ contains $k$ tuples (but does not encode a valid vertex cover), the program backtracks to a previous decision. For this, it determines the largest $\ell$ such that $(\ell, u, v) \in P$ and $(\ell, u) \in C$. It then removes all tuples $(i, u', v') \in P$ with $i > \ell$ and all tuples $(i, u') \in C$ with $i \geq \ell$. It adds the tuple $(\ell, v)$ to $C$. If no such $\ell$ exists, the search tree was traversed completely and no vertex cover of size $k$ exists.
\end{itemize}
Each of these steps is \FOar expressible and the number of steps is bounded by  $\bigO(2^{k+2})$. \qed 
\end{toappendix}

The following lemma %
basically states that every boolean parameterised query can be answered in \ParaS on instances whose domain size is bounded by a function in the parameter.

\begin{lemma}
\label{lem:small}
Let $f \colon \N \to \N$ be a computable function and $(\query,\kappa)$ a
boolean parameterised query with decidable $\query$.
There is a computable function $g$ and a $(g,1)$-parameterised first-order program $(\varphi, \pi)$ that answers $\query$ correctly on instances $\db$ of size at most $f(\kappa(\db))$. 
\end{lemma}
\begin{proof}[Proof idea] We explain the proof idea for input
  structures consisting of a graph $G$ of size $n$ and a parameter
  value $k$ with $n\le f(k)$. 
  The advice $\pi$ produces an advice structure with domain
  $[2^{f(k)^2}]$. It has a ternary relation $E'$ that contains, for every $i\in
  [2^{f(k)^2}]$ all tuples $(i,j_1,j_2)$, for which the $i$-th graph
  over $[f(k)]$ in some canonical enumeration has an edge
  $(j_1,j_2)$. It further contains a unary relation $F$ that contains
  all numbers $i$, for which the $i$-th graph is a yes-instance of
  $\query$.  
  The formula $\varphi$ simply determines with the help of $E'$ and built-in arithmetic %
  the number $i$  of $G$ (as a
  graph over $[n]$) and tests whether $F(i)$ holds. 
\end{proof}
\subparagraph*{Parameterised Dynamic Complexity}

We study parameterised queries in a dynamic context. Formally, a \emph{dynamic parameterised query} $(\query, \kappa, \Delta)$ consists of a parameterised query $(\query, \kappa)$ and a set $\Delta$ of change operations. We say that a parameterised query $(\query, \kappa)$ has an \emph{explicit parameter}, if $\query$ consists of pairs $\inp=(\inp',k)$, where $\inp'$ is a structure, $k$ is a suitably encoded number, and $\kappa(\inp) = k$. %
All concrete parameterised queries we consider in this paper have an explicit parameter.\nilsm{check in the end} %
For example, we often consider the dynamic variant $(\vCover, \Delta_E \cup \pmone)$ of the parameterised vertex cover query, where $\Delta_E\df\{\ins_E, \del_E\}$ and $\pmone \df \{\pone, \mone\}$ denotes the set of  change operations that increment or decrement the given number $k$ by one, as long as $k$ stays in the admissible range. 
So, given some graph $G$ with $n$ vertices, $\pone(G,k) \df (G, k+1)$ if $k < n$, and $\mone(G,k) \df (G,k-1)$ if $k > 1$, and otherwise the changes have no effect. 

For most queries\footnote{The only exception is \knapsack in Section~\ref{subsec:dynprogramming}. } in this paper only parameter values in $\{1,\ldots,n\}$ are meaningful and we only allow such values. %
They can be represented by elements of the domain. 

Similarly as parameterised first-order programs generalise first-order formulas, parameterised dynamic programs extend conventional dynamic programs in two directions: (1) they may use an advice structure whose size depends on the parameter, and (2) they may use first-order programs of parameterised iteration depth. %

A \emph{dynamic program with iteration and advice} is a tuple $(\prog, \pi)$ where $\prog$ is a dynamic program where auxiliary relations are updated with first-order programs  and $\pi$ is an $\advSchema$-advice for an advice schema $\advSchema$. 
For a dynamic parameterised query $(\query, \kappa, \Delta)$, the program~$\prog$ has update rules of the form $\changeRule{\delta(\tpl p)}{S(\tpl x)}{(\Psi_S, \varphi_S)}$  for every $\delta \in \Delta$, where $(\Psi_S, \varphi_S)$ is a first-order program over schema $\inpSchema \cup \auxSchema \cup \advSchema$ such that $\varphi_S$ has the same arity as $S$. States of the program ~$\prog$ are of the form $(\domain \uplus \domain_{\text{adv}}, \inp, \aux, \aux_\text{adv})$ where $\inp$ is the input structure, $\aux$ the auxiliary structure, and $\aux_{\text{adv}}$ is an advice structure over a schema~$\advSchema$. Tuples of the auxiliary structure $\aux$ may range over the domain $\domain \uplus \domain_{\text{adv}}$. 

For two computable functions $f, g: \N \rightarrow \R$, an \emph{$(f, g)$-parameterised dynamic program} is a dynamic program $(\prog, \pi)$ with iteration and advice such that $|\pi(k)| \leq f(k)$ for all $k \in \N$ and all first-order programs of $\prog$ always reach a fixed point after at most $g(\kappa(\inp))$ steps. 
The initial state of such a program depends on an initial  input structure $\inp_0$ and a number $k \in \N$. It is given as $(\domain \cup \domain_{\text{adv}}, \inp_0, \aux_0, \aux^k_{\text{adv}})$ where $\aux^k_{\text{adv}} \df \pi(k)$, $\domain$ and  $\domain_{\text{adv}}$ are the domains of $\inp_0$ and $\pi(k)$, respectively, and $\aux_0$ is an empty $\auxSchema$-structure. 

A dynamic parameterised query $(\query, \kappa, \Delta)$ is maintained by $(\prog, \pi)$ if a distinguished relation $\qRel$ in $\prog_\alpha(\domain \cup \domain_{\text{adv}}, \inp_0, \aux_0, \aux^k_{\text{adv}})$ equals $\query(\alpha(\inp_0))$, for all empty\footnote{For queries with explicit parameter, we require only that in $\inp_0=(\inp'_0,k)$, $\inp'_0$ is empty, but $k$ can be non-zero.} input structures~$\inp_0$, all $k \in \N$, and all sequences $\alpha$ of changes over $\Delta$ such that $\kappa(\alpha'(\inp_0)) \leq k$ for all prefixes $\alpha'$ of $\alpha$.  
So, the dynamic program $(\prog, \pi)$ only needs to maintain $(\query, \kappa, \Delta)$ as long as the parameter value is bounded by the initially given number $k$; nevertheless the program needs to work for arbitrary values of~$k$.
We denote this number $k$ in the following as $\kmax$.

For computable functions $f, g: \N \rightarrow \R$ we define $\ParaSTD(f,g)$ as the class of dynamic parameterised queries that can be maintained by an $(f, g)$-parameterised dynamic program. We define:
\begin{itemize}
\item $\ParaSTD \df \bigcup_{f,g} \ParaSTD(f,g)$,
\item $\ParaSD \df \bigcup_{f} \ParaSTD(f,1)$, 
\item $\ParaTD \df \bigcup_{g} \ParaSTD(0,g)$, 
\end{itemize}

Since the purpose of this article is to explore the basic principles of parameterised, dynamic complexity, we keep the setting simple, in particular with respect to the following two aspects.
First, dynamic programs get a bound $\kmax$ for the parameter values at initialisation time and the program then only needs to deal with changes that obey this parameter bound. 
This ensures that the advice structure does not change throughout the dynamic process. %
Second, we  assume the presence of arithmetic throughout. In non-parameterised dynamic complexity, it is known that under mild assumptions on the query, arithmetic relations  can be constructed by a dynamic program on the fly \cite{DattaKMSZ18}. Similar techniques can be applied for the parameterised setting, yet we ignore this aspect here and assume that $\inp_0 \uplus \pi(k)$ comes with relations  $<$, $+$, and $\times$ over $\domain \uplus \domain_{\text{adv}}$.

For some first intuition we provide a parameterised dynamic program that shows that $(\vCover, \{\ins_E\} \cup \pmone)$ is in \ParaSD %
via the search-tree based approach. This result is not surprising, as it is known that $\vCover \in \ParaS$ \cite{ChenFH17, BannachT18b}. However, the dynamic program for maintaining search trees is conceptually very simple.

\begin{example}\label{ex:vc-tree-dyn}
 We recall the search-tree based parameterised algorithm for \vCover from Example~\ref{ex:vc-tree-fpt}.  The first-order program  of Example~\ref{ex:vc-t} witnesses $\vCover \in \ParaT$ (and thus also in \ParaTD) by constructing a search tree from scratch. In contrast, a dynamic program witnessing $(\vCover, \{\ins_E\} \cup \pmone) \in \ParaSD$ can \emph{maintain} a  search tree. To this end, for a given bound $\kmax$, its advice structure $\aux^{\kmax}_{\text{adv}}$ stores a full binary ``background'' tree $T$ of depth $\kmax$. 
Its auxiliary structure represents the actual search tree $T'$ by maintaining an upward closed set of nodes and the candidate sets of each of those  nodes. As in the search tree algorithm from Example~\ref{ex:vc-tree-fpt}, in every inner node $x$ of $T'$ a branching on the endpoints of some edge $e$ of $G$ is being simulated and in each of $x$'s two children one vertex of $e$  is added to the candidate set. A  node $x$ of $T$ is a leaf of $T'$, if the assigned candidate set of $x$ is an actual vertex cover of $G$ or if $x$ is in level $\kmax$ of $T$. 
The program then only needs to check whether there is a leaf representing a valid vertex cover at a level below the current value of $k$. Maintenance under changes from $\pmone$ is therefore easy.

Maintaining $T'$ under insertion of an edge $(u,v)$ is easy as well: for each leaf of $T'$ that is \emph{not} at level $\kmax$, and whose candidate set does not cover $(u,v)$, the program adds $u$ to the left child and $v$ to the right child (assuming $u<v$). Leaves at level $\kmax$ are not modified, but it might happen that a former vertex cover attached to such a leaf becomes invalid by not covering $(u,v)$.
Maintaining $T'$ under edge \emph{deletions} is slightly more subtle and will be considered  in the proof of Proposition~\ref{prop:vcs}. \qed
\end{example}

 \section{Relationships between Parameterised Classes}\label{section:relations}

In this section we examine how parameterised dynamic and static
complexity classes relate to each other. These relationships are
summarised in Figure~\ref{fig:classes}.  
 
\newcommand{\incomparable}{\not\mathrel{\substack{\subset\vspace{-0.15em}\\=\vspace{-0.15em}\\\supset}}}

\newtcolorbox{staticClassBox}[2][]{colback=teal!5!white, colframe=teal!70, width=3.82cm, left=1mm, top=2mm, bottom=2mm, right=1mm, halign title=flush center, title=#2,fonttitle = \scriptsize, fontupper = \scriptsize, fontlower = \scriptsize #1}
\newtcolorbox{dynamicClassBox}[2][]{colback=blue!5!white, colframe=blue!70, width=3.82cm, left=1mm, top=2mm, bottom=2mm, right=1mm, halign title=flush center, title=#2,fonttitle = \scriptsize, fontupper = \scriptsize, fontlower = \scriptsize #1}

     \begin{figure}[t]
        \centering
          \begin{tikzpicture}[
            very thick,
            xscale=0.9,
            yscale=.7,
            classBoxNC/.style = {rectangle, rounded corners = 5pt, minimum width=3.6cm, minimum height = 0.5cm, color=white, font=\scriptsize},
        staticClassBoxNC/.style = {classBoxNC, fill=teal!70},
        dynamicClassBoxNC/.style = {classBoxNC, fill=blue!70},
        inclusion/.style = {shorten <=1pt, shorten >=1pt},
        incomparable/.style = {inclusion, dashed},
        notSubset/.style = {-stealth, dotted, inclusion}
          ]
                         
          \node[staticClassBoxNC] (FPT) at (0,5.2) {\FPT};    
          \node[dynamicClassBoxNC] (STD) at (0,4) {\ParaSTD};
          
          \node (SD) at (-5.0,2.5) {
            \begin{dynamicClassBox}{\ParaSD}
              \begin{itemize}
                \item $\knapsack$
                \item $\longpath$
              \end{itemize}
            \end{dynamicClassBox}
          };      

          \node (TD) at (5,2.5) {
            \begin{dynamicClassBox}{\ParaTD}
            \begin{itemize}
                \item \fVS
             \end{itemize}
            \end{dynamicClassBox}
          }; 

          \node (ST) at (0,0) {
            \begin{staticClassBox}{\ParaST}
              \begin{itemize}
               \item \st{\fVS}
             \end{itemize}
            \end{staticClassBox}
          }; 
            
          \node (S) at (-5,-2.5) {  
            \begin{staticClassBox}{\ParaS}
              \begin{itemize}
              \item \vCover
  \item $\{(Q, \kappa)\mid \kappa(x)=|x|,$\\ \mbox{ }\hfill $ Q \text{ decidable}\}$
  \end{itemize}
             \tcblower
              \begin{itemize}
                \item \st{\longpath}
             \end{itemize}
            \end{staticClassBox}
          };
            
          \node (T) at (5,-2.5) {
            \begin{staticClassBox}{\ParaT}
              \begin{itemize}
                \item \vCover
                \item \cString
                \item \longpath
              \end{itemize}
            \end{staticClassBox}
          };
          
          \node (PSPACE) at (5,5) {
            \begin{staticClassBox}{}
             \begin{itemize}
              \item\mbox{$\{(Q,\kappa) \mid Q \in \PSPACE\}$}
             \end{itemize}
            \end{staticClassBox}
          };

          \draw[inclusion] (TD) -- (PSPACE);    
          \draw[incomparable](S.south) -| (-5,-5) -| (8,-5) |- node [pos = 0.4, shift=(180:0.3cm)] {$\incomparable$} (PSPACE);   
            
            \draw[inclusion] (STD) -- (FPT);
            \draw[inclusion] (TD) -- node [shift=(45:0.3cm)] {$\neq$}(STD);
            \draw[inclusion] (SD) -- (STD);
            \draw[inclusion] (T) -- node [shift=(45:0.3cm)] {$\neq$} (ST);
            \draw[inclusion] (S) -- node [shift=(125:0.3cm)] {$\neq$}(ST);
            \draw[inclusion] (ST) -- node [pos=0.2, shift=(0:0.3cm)] {$\neq$} (STD);
            \draw[inclusion] (S) -- node [shift=(0:0.3cm)] {$\neq$} (SD);
            \draw[inclusion] (T) -- node [shift=(0:0.3cm)] {$\neq$}(TD);
            \draw[notSubset] (SD) -- node [pos=0.7, shift=(90:0.3cm)] {$\not\subseteq$}(TD);
            \draw[notSubset] (SD) -- node [pos=0.4, shift=(45:0.3cm)] {$\not\subseteq$}(ST);            
            \draw[incomparable] (TD) -- node [pos=0.4, shift=(135:0.3cm)] {$\incomparable$} (ST);
            \draw[incomparable] (T) -- node [pos=0.2,shift=(90:0.3cm)] {$\incomparable$} (S);
                        \draw[incomparable, bend left = 27] (TD) edge node [pos=0.2,shift=(135:0.3cm)] {$\incomparable$} (S);
          \end{tikzpicture}
        \caption{
            Inclusion diagram of the main classes. Solid lines indicate %
            inclusions. Dashed lines marked with $\incomparable$ indicate that the two classes are incomparable. A directed, dotted edge marked with $\not\subseteq$ from  $\calC$ to  $\calC'$ indicates $\calC \setminus \calC' \neq \emptyset$. 
			If $\calC$ is a dynamic class and $\calC'$ a
                        static class, $\calC \subseteq \calC'$ means that for each $(Q, \kappa, \Delta) \in \calC$ with exhaustive $\Delta$ it holds that $(Q, \kappa) \in \calC'$, and $\calC' \subseteq \calC$ means that for each $(Q, \kappa) \in \calC'$ it holds that $(Q, \kappa, \Delta) \in \calC$, for arbitrary $\Delta$.         
            \label{fig:classes}
          }
       \end{figure}
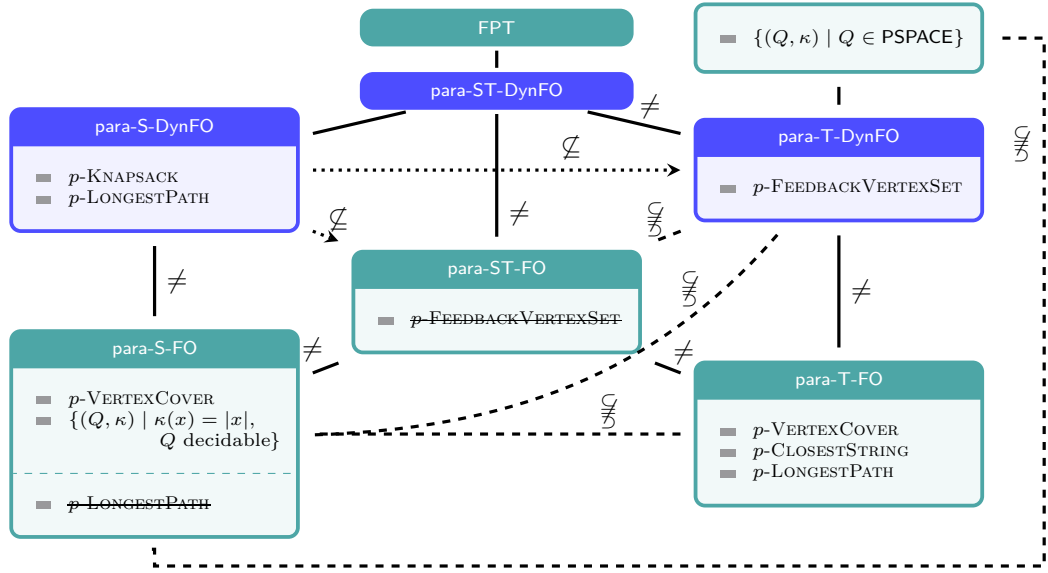

As a sanity check,  we show first that 
every parameterised query $(Q, \kappa)$ with $(Q, \kappa, \Delta) \in
\ParaSTD$ is in \FPT.  For queries in
\ParaTD the respective algorithm only needs polynomial space.
Both statements require that $\Delta$ is \emph{exhaustive}, i.e., that
it contains the single-tuple insertion operation
$\ins_R$ for every input relations $R$. This ensures that 
every possible input structure for $Q$ can be obtained by a change
sequence.\footnote{Clearly, a more
  general definition would be possible here, but we avoid that in the
interest of simplicity.} 
 
     \begin{propositionrep}\label{proposition:classes:upper}
       \begin{enumerate}[(a)]
       \item For every $(Q, \kappa, \Delta) \in \ParaSTD$ with exhaustive $\Delta$ it holds that $(Q, \kappa) \in \FPT$.
       \item For every $(Q,\kappa,\Delta) \in \ParaTD$ with exhaustive
         $\Delta$, the parameterised query $(Q, \kappa)$ can be
         solved by an \FPT-algorithm that uses at most polynomial space with respect to the input size.  In particular, $Q\in\PSPACE$.
       \end{enumerate}
     \end{propositionrep}
     \begin{appendixproof}
       We only sketch the proof.
\begin{enumerate}
 \item Let $(Q,\kappa,\Delta)$ be a parameterised query in $\ParaSTD(f,g)$, as witnessed by some parameterised dynamic program $\prog$ with advice $\pi$ obeying $|\pi(k)| \leq f(k)$ and update programs with iteration depth $g(k)$, for some computable functions $f, g$. Given an instance $\inp$ of size $n$, an \FPT-algorithm for $(Q, \kappa)$ first computes $\pi(\kappa(\inp))$ and then simulates $\prog$ for a sequence of insertions that constructs $\inp$ from an initially empty structure. In each one of these polynomially many update steps the algorithm needs to evaluate a first-order program for each auxiliary relation. Each of the $g(k)$ iteration steps in these evaluations requires $(n + f(\kappa(\inp)))^c$ time for some constant $c \in \N$. All in all this yields an \FPT-running time. 
 \item In the case of $\ParaTD$, the \FPT-algorithm sketched in  (a) essentially only needs to store the current auxiliary relations at any point in time, which amounts to polynomial space in $n$. %
\qedhere
\end{enumerate}
\end{appendixproof}  

Statement (b) does not hold for parameterised classes with advice,
as we formalise with the next proposition, which is an immediate consequence of Lemma \ref{lem:small}.

\begin{proposition}\label{prop:paras-n}
Every parameterised query $(Q, \kappa)$ with decidable $Q$ and $\kappa(x)=|x|$ is in $\ParaS$.
\end{proposition}
\begin{proposition}\label{prop:SDynS}
 For any $(Q, \kappa) \in \ParaS$ and any
 $\Delta\subseteq\Delta_\inpSchema$ (or $\Delta\subseteq\Delta_\inpSchema\cup\pmone$) it holds that
 $(Q,\kappa,\Delta) \in \ParaSD$. %
\end{proposition}
\begin{proofsketch}
  Let $(Q, \kappa)\in\ParaS$ by some $(f,1)$-parameterised \FO
  program \FOprog. In principle, a parameterised dynamic program can simulate
  \FOprog from scratch after each change. However, since
  the parameter of \inp might change, it might need different advice
  structures from \FOprog. However, there is an easy solution for
  this.  For the given \kmax, the dynamic program gets as its advice \emph{all} advice
  structures $\pi(1),\ldots,\pi(\kmax)$ of \FOprog.
\end{proofsketch}
The same argument can be applied for $\ParaST$ and $\ParaSTD$.

In addition to the above inclusions and those that are immediate from
the definitions, we observe the following separations between
parameterised classes (also see Figure~\ref{fig:classes}). Some proofs
are deferred to the next section.  

\begin{proposition}\label{proposition:classes:lower}
\begin{enumerate}
 \item There is a $(Q, \kappa) \in \ParaS$ such that $(Q, \kappa, \Delta) \not\in \ParaTD$, for any exhaustive $\Delta$.
 \item There is a $(Q, \kappa) \in \ParaT$ such that $(Q, \kappa) \not\in \ParaS$.
 \item There is a $(Q, \kappa, \Delta) \in \ParaTD$ with exhaustive $\Delta$ such that $(Q, \kappa) \not\in \ParaST$.
 \item There is a $(Q, \kappa, \Delta) \in \ParaSD$ with exhaustive $\Delta$ such that $(Q, \kappa) \not\in \ParaST$.
\end{enumerate}
\end{proposition}
\begin{proofsketch}
  Part (a) is a consequence of Proposition
  \ref{proposition:classes:upper} and Proposition \ref{prop:paras-n},
  and witnessed by any parameterised problem $(Q,\kappa)$ with decidable $Q
  \not\in \PSPACE$ and $\kappa(x) = |x|$. Part (b) is witnessed by the
  problem \longpath which is not in $\ParaS$ \cite{BannachST15}, but
  in $\ParaT$ as we will see in Proposition \ref{prop:longpath}. For
  (c)  we observe that \fVS is not in \ParaST, as otherwise the
  restriction to inputs with parameter $k=0$ would yield a first-order
  formula that expresses acyclicity of undirected graphs. In
  Proposition~\ref{prop:fVSparaTD} we will show that $(\fVS, \Delta_E \cup \pmone)$ is in
  \ParaTD. The separation %
  for (d) can be shown
  with  the help of  connectivity of undirected graphs. To this end, we
  consider the parameterisation by the maximal node degree. It is
  well-known that even for fixed $k=2$ this property is not
  expressible in \FOar, see \cite{FurstSS84}, and thus it is not in $\ParaST$. On
  the other hand, towards (d), the unparameterised version is in \DynFO and thus
  the parameterised version is in $\ParaSD$.\footnote{Of course,
    this argument could have been used for (c) as well, but there we prefer
    a more ``natural'' parameterisation.}
\end{proofsketch}

 \section{Methods for Parameterised Complexity}\label{section:methods}

The goal of this section is to explore the
transferability of known methods
from the realm of parameterised algorithms to
dynamic parameterised complexity.
We are thus not always interested in ``best algorithms'' but rather
want to exemplify how sequential algorithmic methods for static problems translate into the
dynamic (highly parallel) setting.

We start by describing colour-coding, since it turns out as particularly useful in the dynamic
context and we use it in many other subsections. Then we consider three classical
methods for parameterised algorithms, bounded search trees,
kernelisation and dynamic programming. 
Afterwards we give an example for the iterated compression method, which uses an adaption of a technique from dynamic complexity.

\subsection{Colour-Coding}\label{subsec:colour-coding}

In this subsection, we establish the usefulness of the colour-coding technique, as presented in~\cite{AlonYZ95}, in
our setting by a concrete example, \longpath.
\paraProblemDescription{\longpath}{An undirected graph $G=(V,E)$, $s,t \in V$ and $\ell \in \N$}{$\ell$}{Is there a (simple) path from $s$ to $t$ of length $\ell$?}
This problem can be solved with the help of %
\emph{universal colouring families}. Such a
family is a small set of functions that map nodes to colours such that
if a path of length $\ell$ exists, one of these functions %
colours the nodes of the path with a fixed sequence of $\ell+1$ colours.
A parallel algorithm for \longpath
therefore only needs to test in parallel, for each function of a
universal colouring family, whether it produces such a coloured path from
$s$ to~$t$. 

More precisely, a \emph{$(n,k,c)$-universal colouring family}
$\Lambda$ has, for every subset $S\subseteq[n]$ of size $k$ and for every
  mapping $\mu: S \to [c]$,   at least one function
  $\lambda \in \Lambda$ with $\lambda(s)=\mu(s)$, for
  every~$s\in S$.
In \cite[Theorem 3.2]{BannachST15} a family
$\Lambda_{n,k,c}$ of
such functions is defined. The definition can be found in the
appendix.
\begin{toappendix}
 A {$(n,k,c)$-universal colouring family} $\Lambda_{n,k,c}$ can be
 constructed as follows.
  \begin{align*}
    \lambda_{p,j} (x) & \: \df \:  (jx \bmod p) \bmod k^2,\\
    \Lambda'_{n,k} & \: \df \:  \{\lambda_{p,j} \mid \text{$p$ prime},
                     p<k^2\log n, j\in[p-1]\},\\
    \Lambda_{n,k,c} & \: \df \:  \{\omega\circ \lambda_{p,j} \mid \omega:\{0,
                      \ldots,k^2 -1\}\to [c], \lambda_{p,j} \in \Lambda'_{n,k}\}.
  \end{align*}
\tsm{Describe how these functions are FO-definable.}
\end{toappendix}
In the presence of arithmetic, these functions are easily first-order
definable and can be enumerated in a first-order fashion.

\begin{propositionrep}\label{prop:longpath}
  \begin{enumerate}[(a)]
   \item $\longpath\in\ParaSD$.
   \item \label{item:longpath-paraT} 
   $\longpath\in\ParaT$.
\end{enumerate}
\end{propositionrep}
\begin{proofsketch}
  In both parts of the proof, we use the colour-coding approach as sketched above. For a graph $G$,
a colouring function $\lambda$, and a set $C$ of colours, a
\emph{$C$-coloured path  under $\lambda$} is a path
whose nodes are mapped to $C$ in a one-one fashion by $\lambda$.  

  For solving the \longpath problem with parameter $\ell$, we consider 
 the $(n,k,k)$-universal colouring family $\Lambda \df \Lambda_{n,k,k}$ with $k \df \ell +
 1$. Then  a graph has a simple path of length $\ell$ from $s$ to $t$ if
 and only if there is a $[k]$-coloured path
 from $s$ to $t$ under some $\lambda\in\Lambda$.

  We first show $\longpath\in\ParaSD$. The dynamic program uses a
  dynamic programming approach (in the classical sense of this
  term). It stores, for each $\lambda \in \Lambda$ and each pair $(u,
  v)$ of nodes, the set $\calC$ of color sets $C$, for which there is
  a $C$-coloured path from $u$ to $v$ under $\lambda$. 

  That $\longpath\in\ParaT$ can be shown with the help of the same
    universal colouring family $\Lambda$ as above, which consists of $f(k) \textsf{poly}(n)$ colourings. The idea for the
    program is to test, in $f(k)$ iterations and in each iteration for $\textsf{poly}(n)$ colourings in parallel, whether
  there is a $[k]$-coloured path from $s$ to $t$ under the current colouring.
    A suitably coloured path can be found in $k$ iterations.
  More details can
  be found in the appendix.
\end{proofsketch}
\begin{appendixproof} 
We give a more detailed description of (a). We recall that the program stores, for each $\lambda \in \Lambda$ and each pair $(u,
  v)$ of nodes, the set $\calC$ of color sets $C$, for which there is
  a $C$-coloured path from $u$ to $v$ under $\lambda$.
  At
  initialisation time, this
  information is easy to compute in a first-order fashion, since the
  initial graph is  empty. We will see
  that it can also be easily  updated. 

  Let us first state more precisely what
  is stored by the dynamic program. In order to be able to address
  subsets of $C$ as well as all $\lambda \in \Lambda$, the domain
  $\domain_{\text{adv}}$ of the advice structure is chosen as
  $[k^{k^2}]$. Subsets $C_j$ of $[k]$ are encoded by numbers $j$ in
  $[2^{k}]$,  where $i\in C_j$ if and only if the $i$-th bit of $j$ is
  $1$. All functions $\omega:\{0,\ldots,k^2-1\}\to [k]$ are stored in
  a ternary relation $L \subseteq [k^{k^2}]\times[k^2]\times[k]$ of
  the advice structure of the program where $(i,m,c)$ is in $L$ if the
  $i$-th function maps $m-1$ to $c$.  For encoding the colouring functions, recall that each $\lambda \in \Lambda$ is of the form $\omega\circ \lambda_{p,j}$, where $p \in [k^2\log n]$, $j \in [p-1]$, and $\omega:\{0, \ldots,k^2 -1\}\to [k]$. Thus each $\lambda \in \Lambda$ can be addressed by a tuple of elements from $\domain$ and $\domain_{\text{adv}}$. 

  As auxiliary relations, the dynamic program stores, for
  each~$\lambda \in \Lambda$, a relation $R_\lambda \subseteq
  [2^{k}]\times V \times V$ with the intention that $(C,u,v)$ is in
  $R_\lambda$ if there is a $C$-coloured path from $u$ to $v$.  
  Of course the program cannot store $|\Lambda|$ many auxiliary relations,
  since this number is not a constant. Therefore, it uses a relation
  $\widehat{R}$ that represents all relations $R_\lambda$,  with the
  help of additional components $\tpl a$ (consisting of elements from $\domain_{\text{adv}}$) and $\tpl p$ (consisting of elements from $\domain$) for addressing each~$\lambda$. Given $\widehat{R}$, an $\FO$-formula can easily extract pairs of nodes $(u,v)$ that are connected by a simple path of length $\ell$.

  We argue that each $\widehat{R}$ can be updated by a first-order
  program.  After inserting an edge $(u, v)$, a tuple $(C,a,b)$ is in
  $R_\lambda$   if it was in $R_\lambda$ already before the insertion or if $C = C' \uplus C''$ for some $C'$ and $C''$ such that
  $(C', a, u) \in R_\lambda$ and $(C'', v, b) \in R_\lambda$. After
  deleting an edge $(u, v)$, the relation $R_\lambda$ is updated as
  follows. Suppose $u$ and $v$ are coloured $c_u$ and $c_v$ under $\lambda$. If $c_u$ or $c_v$ is not in $C$, the status of
  $(C,a,b)$ in $R_\lambda$ does not change.
 Otherwise, $(C, a, b)$ is in the updated relation if there is
  an edge $(u',v')$ with $u' \neq u$ or $v' \neq v$ such that the
  colours of $u'$ and $v'$ are $c_u$ and $c_v$, respectively, and
  there are sets $C',C''$ such that  $(C', a, u'), (C'', v', b) \in R_\lambda$ and
  $C = C' \uplus C''$.

We next show (b), that is, that $\longpath\in\ParaT$ with the help of the same universal
colouring family $\Lambda$ as above. The idea for the program is to
test, in parallel for all $\lambda_{p, j}$, whether there is a
function  $\omega:\{0,\ldots,k^2-1\}\to [k]$ for which there is a
$[k]$-coloured path from $s$ to $t$ under the colouring
$\lambda=\omega\circ\lambda_{p,j}$.  

To this end the program cycles through all possible functions $\omega:\{0,\ldots,k^2-1\}\to [k]$. It uses a ternary relation $\Omega$ that stores a triple $(a, b,c)$ if the current function $\omega$ maps $(a, b)$ to $c$. The lexicographically next function can be defined by a first-order formula using the presence of arithmetic. For testing whether the graph with colouring $\lambda \df \omega\circ\lambda_{p,j}$ has a
$[k]$-coloured path from $s$ to $t$, the program cycles through all permutations $\pi$ of $[k]$ and  computes, for each $i$, the set of nodes that can be reached from $s$ by a path using the colours $\pi(1),\ldots,\pi(i)$, in that order. The second part can be achieved easily using an additional binary relation that is intended for storing tuples $(i, a)$ if $a$ can be reached by a path $\pi(1),\ldots,\pi(i)$, and observing that this relation can be iteratively computed by a first-order formula. For cycling through all permutations, the program actually cycles through all functions $\pi: [k] \to [k]$ and tests whether $\pi$ is indeed a permutation. Again, this can be achieved by using additional relations and suitable first-order formulas.

The number of required first-order-iterations to run this algorithm is
bounded by $\bigO(k^{k^2}k!k)$. 
\end{appendixproof}

\subsection{Bounded-depth search trees}\label{subsec:search-tree}

Bounded-depth search trees are a classical technique in parameterised
complexity. Already in Example \ref{ex:vc-tree-dyn} we outlined that
search trees are a viable tool also in the dynamic context by showing
how a search tree for $\vCover$ can be maintained under edge
insertions. Here we provide more examples. First we  extend Example
\ref{ex:vc-tree-dyn} towards edge deletions. Afterwards we
consider two further problems, for which the known search-tree based
algorithms can be adapted to place them  in $\ParaT$ or $\ParaTD$,
respectively: \cString and \fVS. Although we conjecture that these
problems are also in $\ParaSD$, we were not able to prove
it. %

\begin{proposition}\label{prop:vcs}
    $(\vCover, \Delta_E \cup \pmone)\in
\ParaSD$ by a search-tree-based dynamic program.
\end{proposition}
\begin{proof}[Proof sketch]
    Let $T$ and $T'$ be defined as in Example~\ref{ex:vc-tree-dyn}. It
    remains to explain how edge deletions can be handled.
    If an edge $(u,v)$ is deleted, and a node $x$ of $T'$ used
    $(u,v)$ for its branching step, the induced subtree of $x$ can be
    replaced by the induced subtree of its left child $y$, see
    Figure~\ref{fig:example:vCover}.\footnote{Of course, the right
    child would work equally well.} More precisely, the children $u'$ and $v'$ of $y$ become the new children of $x$, and in all candidate sets below $u'$ and $v'$ the vertex $u$ is removed.

    \tikzset{
        treeEdgeLabel/.style={draw=none,fill=none},
        treeEdgeLabelLeft/.style={treeEdgeLabel,above left},
        treeEdgeLabelRight/.style={treeEdgeLabel,above right},
        subtree/.style={
            draw,solid,shape border uses incircle,
            isosceles triangle,shape border rotate=90,yshift=-3.35em
        }
    }

   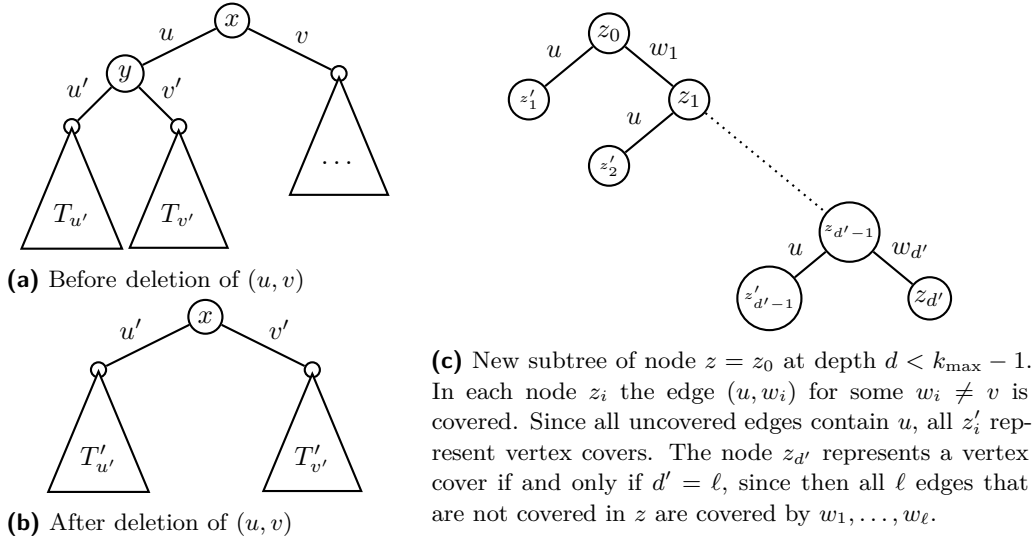
\begin{figure}
     \centering
     \begin{minipage}{.39\linewidth}
       \begin{subfigure}{.95\linewidth}
       \centering
\begin{tikzpicture}[nodes={circle,draw=black,fill=white,inner sep=2pt},thick,level distance=2em,
    level 1/.style={sibling distance=8em},
    level 2/.style={sibling distance=4em}]
    \node (i) {$x$}
        child{
            node (2i) {$y$} 
                child{
                    node (4i) {}
                    node[subtree] (4is) {$T_{u'}$}
                    edge from parent node[treeEdgeLabelLeft] {$u'$}
                }
                child{
                    node (4i2) {}
                    node[subtree] (4is) {$T_{v'}$}
                    edge from parent node[treeEdgeLabelRight] {$v'$}
                }
            edge from parent node[treeEdgeLabelLeft] {$u$}
        }
        child{
            node (2i1) {}
            node[subtree] (2i1s) {$\dots$}
            edge from parent node[treeEdgeLabelRight] {$v$}
        }
        ;
\end{tikzpicture}          \subcaption{Before deletion of $(u,v)$}
       \end{subfigure}\\
       \begin{subfigure}{.95\linewidth}
       \centering
\begin{tikzpicture}[nodes={circle,draw=black,fill=white,inner sep=2pt},thick,level distance=2em,
    level 1/.style={sibling distance=8em},
    level 2/.style={sibling distance=4em}]
    \node (i) {$x$}
        child{
            node (4i) {}
            node[subtree] (4is) {$T'_{u'}$}
            edge from parent node[treeEdgeLabelLeft] {$u'$}
        }
        child{
            node (4i2) {}
            node[subtree] (4is) {$T'_{v'}$}
            edge from parent node[treeEdgeLabelRight] {$v'$}
        }
        ;
\end{tikzpicture}          \subcaption{After deletion of $(u,v)$}
       \end{subfigure}
     \end{minipage}
     \begin{minipage}{.59\linewidth}
        \begin{subfigure}{.95\linewidth}
            \centering
\begin{tikzpicture}[nodes={circle,solid,draw=black,fill=white,inner sep=1pt,minimum size=1.5em},thick,level distance=2.5em,
    level 1/.style={sibling distance=6em},
    level 2/.style={sibling distance=6em},
    level 3/.style={sibling distance=6em}]
    \node (z) {$z_0$}
        child{
            node[] {\tiny$z_1'$}
            edge from parent node[treeEdgeLabelLeft] {$u$}
        }
           child{
              node (z1) {$z_1$}
              child{
                node {\tiny$z_{2}'$}
                  edge from parent[solid] node[treeEdgeLabelLeft] {$u$}
              }
        child{
            child{
                edge from parent[draw=white]
            }
               child{
                    node(zi1) {\tiny$z_{d'-1}$}
                    child {
                        node {\tiny$z_{d'-1}'$}
                        edge from parent[solid] node[treeEdgeLabelLeft] {$u$}
                    }
                        child{
                            node(zmax) {$z_{d'}$}
                            edge from parent[solid] node[treeEdgeLabelRight] {$w_{d'}$}
                        }
                    edge from parent[dotted] 
                }
                 edge from parent[dotted] 
             }
            edge from parent node[treeEdgeLabelRight] {$w_{1}$}
        }
        ;
\end{tikzpicture}
             \subcaption{
                New subtree of node $z=z_0$  at depth $d<\kmax-1$.
                In each node $z_i$ the edge $(u, w_i)$ for some $w_i \neq v$ is covered.
                Since all uncovered edges contain $u$, all $z_i'$
                represent  vertex covers. 
                The node $z_{d'}$ represents a vertex cover if and
                only if
                $d'=\ell$, 
                since then all $\ell$ edges that are not covered in $z$ are covered by  $w_1, \ldots, w_{\ell}$.
            }
        \end{subfigure}
     \end{minipage}

        \caption{
            Modification of the search tree for \vCover after deletion of an edge $(u,v)$.
            The new sub-trees $T'_{u'}$, $T'_{v'}$ of $x$ are obtained from $T_{u'}$, $T_{v'}$ respectively,
            by adding two new children to leaves that do not represent a vertex cover.
        }\label{fig:example:vCover}
    \end{figure}

        The subtree of $x$ might now (1) have leaves of depth $\kmax-1$ that do not represent an actual vertex cover, since the modification reduces the
    depth of all nodes in the subtree of $x$, and (2) have leaves at a smaller depth $d<\kmax{-}1$ which do not represent a vertex cover, since $u$ is removed from the candidate sets and thus edges adjacent to $u$ may not be covered any more. These defects can be corrected successively. 
    
    First, for each of the leaves from (1), two new children are added, with the help of the
    lexicographically smallest uncovered edge $(u'',v'')$. 

    Regarding a leave $z$ with property (2), observe that its candidate set can miss only
    edges of the form $(u,w)$, where $w\not=v$. It is easy to see that
    the subtree rooted at $z$ can be chosen in the following shape. 
    Let $W=\{w_1, \ldots, w_{\ell}\}$ be the set of vertices with an uncovered edge $(u,w_i),i\in [\ell]$.
    The new subtree having depth $d' = \min\{\ell,\kmax-d\}$ consists of a path with nodes
    $z_0,\ldots,z_{d'}$ such that $z_0=z$ and for each $i\ge 0$, the left child of
    $z_i$ is a leaf obtained by adding $u$ to the candidate set and for the
    right child $z_{i+1}$, $w_{i+1}$ is added to the candidate set. %

    This new subtree can be defined in a first-order fashion with the help of colour coding.
    Let $U$ be the candidate set of $z$.
    Then $W$ consists of all neighbours of $u$ that are not in $U$, so $W$ is easily FO-definable.
    To define the subtree, $d'$ vertices have to be chosen from $W$.
    To this end, we consider colourings of $W$ that map $W$ to $[\ell]$.
    With the help of an $(n,\kmax,\kmax)$-universal colouring family,
    one can quantify over such colourings and by picking (a canonical) one,
    the new subtree can be defined by choosing each $w_i$ as  the node coloured with $i$,
    for every $i\in[d']$. All these updates can be expressed by first-order formulas.
\end{proof}

For the closest string problem, we fix an alphabet $\Sigma$, and let $\hamD(s_1,s_2)$ denote the Hamming distance of
$s_1$ and $s_2$, i.e.~the number of positions where $s_1$ and $s_2$ differ.

\paraProblemDescription{\cString}{Strings $s_1, \ldots, s_n \in \Sigma^n$ for some $n \in \N$, and $d \in \N$}{$d$}{Is there a string $s \in \Sigma^n$ such that $\hamD(s,s_i) \leq d$?}

An input to \cString with strings of length $n$ is represented by a
structure with domain $[n]$.  It has the  natural linear order on $[n]$ and, for every $\sigma \in \Sigma$ 
a relation $R_{\sigma}(i,j)$ with the meaning $s_i[j] = \sigma$, i.e.
string $s_i$ has symbol $\sigma$ at position $j$.

A search tree (see \cite[Section 8.5]{Niedermeier06}) of depth at most $d$ and degree at most $d+1$ gradually adapts a candidate string $s$, which is initially set to $s_1$. If an input string $s_i$ is ``far apart'' from $s$, the tree branches on the first $d+1$ differences and changes $s$ towards $s_i$. 

\begin{propositionrep}
    $\cString \in \ParaT$.
  \end{propositionrep}
  The construction is quite straightforward and can be found in the appendix.
  \begin{appendixproof}
    We first recall the classical static \FPT algorithm for \cString based on bounded-depth search trees \cite[Section 8.5]{Niedermeier06}.
It uses the following observations. Firstly, as necessarily $\hamD(s,s_1) \leq d$, if a solution string $s$ exists it can be derived from $s_1$ by changing at most $d$ positions of $s_1$. Secondly, if $\hamD(s_1,s_i) \leq d$ for all $i \in \{2, \ldots, n\}$, then $s_1$ is a solution string. Otherwise there is an $i$ such that $\hamD(s_1,s_i) \geq d+1$, and a solution $s$ needs to agree with $s_i$ on at least one of the first $d+1$ positions that $s_1$ and $s_i$ differ on.

Based on these observations, a search tree is constructed as follows. Every node $v$ is labelled with a candidate string $s_v$ and its depth $d_v$. For the root node $r$ we set $s_r = s_1$.
For each tree node $v$ of depth $d_v < d$ such that $s_v$ is not a solution string for the instance $s_2, \ldots, s_n$, a string $s_i$ is selected such that $\hamD(s_v,s_i) \geq d+1$. Let $j_1, \ldots, j_{d+1}$ be the first $d+1$ positions in which $s_v$ and $s_i$ differ.
For each $j \in \{j_1, \ldots, j_{d+1}\}$ a child of $v$ is added that
is labelled with the string that results from $s_v$ by replacing the
symbol at position $j$ with $s_i[j]$.

We now adapt the classical search tree approach from above, analogously to Example~\ref{ex:vc-t}, and construct a first-order program \FOprog that traverses the search tree in a depth-first manner.

The program uses a relation $C$ that represents a path in the search tree to the current node. The relation $C$ contains a tuple $(\ell,i,j,m)$ if at depth $\ell$ of the search tree the string $s_i$ has hamming distance at least $d+1$ from the current candidate string, the position $j$ is the $m$-th position that the candidate string and $s_i$ differ on, and the new candidate string is obtained by replacing the symbol at position $j$ by $s_i[j]$. 

Note that one can define the current candidate string in \FO given $C$. The set of strings that have hamming distance at least $d+1$ to the candidate can be defined in \ParaT, as one can count in $d+1$ iterations the number of differences up to $d+1$ for each string.
Therefore, first-order formulae can check whether the current relation $C$ encodes a solution string, simulate the move to a child node of the current search tree node if its depth is smaller than $d$, or otherwise simulate a backtrack step, if the search tree is not already fully traversed.

The move to a child is performed as follows, assuming that $C$ contains $\ell < d$ tuples. Let $i$ be minimal such that $s_i$ differs on at least $d+1$ positions with the current candidate string, and let $j$ be the first of them. Then the tuple $(\ell+1,i,j,1)$ is inserted into $C$.

If $C$ contains $d$ tuples and a backtrack step needs to be performed, let $\ell$ be the largest number such that $(\ell,i,j,m) \in C$ with $m \leq d$.
All tuples $(\ell',i',j',m')$ with $\ell' \geq \ell$ are removed, and the tuple $(\ell,i,j'',m+1)$ is inserted into $C$, where $j''$ is the first position after position $j$ such that $s_i$ has a different symbol at that position than the candidate string that is defined by the first $\ell-1$ tuples of $C$.
\end{appendixproof}

Next, we explore the parameterised problem \fVS.
Given a graph $G=(V,E)$, a feedback vertex set (\fvsSet) for $G$
is a set $S \subseteq V$ such that for every cycle $C$ in~$G$, $S \cap
C \neq \emptyset$ holds,
i.e. $G-S$ is a forest. 

\paraProblemDescription{\fVS}{An undirected graph $G$}
    {$k$}{Does $G$ have a feedback vertex set  of size $k$?}

\begin{propositionrep}\label{prop:fVSparaTD}
    $(\fVS, \Delta_{E} \cup \pmone) \in \ParaTD$.
  \end{propositionrep}
  \begin{inlineproof}[Proof idea]
    We show that $\fVS$ can be maintained in $\ParaTD$ using
a depth-bounded search tree, similarly as for \vCover. The result
uses a well-known approach relying on the fact that if a graph of minimum degree 3 has
a \fvsSet of size $k$ then the length of its minimal cycle is bounded
by $2k$ (e.g. \cite{Downey1995}). A branching step consists of two
phases:  removing vertices of degree 1
or 2, and finding a small cycle. Then, each branch selects one of
these cycle vertices for the \fvsSet candidate. 
At the leaves of the search tree it has to be checked if the graph
obtained by deleting the chosen vertices of the current branch is
acyclic. A cycle exists, if there exists an edge $(u,v)$ and
$u$ is reachable from $v$ in $G-(u,v)$, thus this can be decided with
the transitive closure of the edge relation. The latter can be maintained in DynFO under edge insertions and deletions \cite{DattaKMSZ18}
and, as we show  in the appendix,
also under vertex deletions (simulated by removing all edges of a
vertex).
  \end{inlineproof}
\begin{appendixproof}
    We describe a parameterised, dynamic program \prog with iteration but without advice
    that maintains $\fVS$ under single edge changes and parameter changes by 1.
    On each edge change it follows the procedure of the program we sketched in \Fref{ex:vc-t},
    that is, it iterates over the search tree in a depth first manner.
    Similar to \Fref{ex:vc-t} the program stores a representation of the path
    from the root of the search tree to the current node and keeps
    track of the candidate set $F$ and all cycles that were used for branching. 
   Additionally it stores the graph $G-F$ and all auxiliary relations
   of the program \prog' of \Fref{prop:reach_complex_changes} to
   maintain reachability on $G-F$, for all
   nodes along the current branch. 
    In order to have the latter available for the root of the search tree
    \prog simulates \prog' for all edge changes on the input graph $G$.
    Performing a backtracking step can be done similarly as in \Fref{ex:vc-t}:
    Find the largest depth $\ell$ such that the vertex added into the candidate set at depth $\ell$
    was not the largest of the stored cycle and use the next vertex of
    the cycle.
    Deciding if at any step the candidate set $F$ is a valid \fvsSet,
    i.e.~checking if there is an edge $(u,v)$ in $G-F$ such that $u$
    is reachable from $v$ in $(G-F)-(u,v)$, amounts to simulating one
    further step of \prog' for a single edge deletion. Since the
    size of the search tree is bounded by $(2k)^{k+1}$, it suffices to
    show that each branching step can be done by a number of
    iterations that is bounded by a function in $k$.
    
    We next describe how a single branching step is done, in principle.
    Let $G_0\df G-F$ be the graph of the current search tree node and
    let $k_0\df k-|F|$.
    The program \prog performs the following three steps.
    \begin{enumerate}[(1)]
        \item Get rid of vertices of degree 1 by removing maximal (attached) trees.
            A vertex $u$ is part of an \emph{attached tree},
            if there is a vertex $v$ so that, in $G_0-(u,v)$,
            $u$ is not reachable from $v$ and the connected component of $u$ is a tree.
        \item Get rid of vertices of degree 2 by merging simple paths
            whose inner vertices all have degree 2 with a single edge between its endpoints.
        \item Search for a cycle of length $2k_0$ and conclude that
          there is no $\fvsSet$ of size $k$ containing~$F$,            if there is no such cycle.
    \end{enumerate}

    We note, that the graph $G_2$ resulting from steps (1) and (2) has
    a \fvsSet of size $k_0$  if and only if $G$ has a \fvsSet of size
    $k$ that contains $F$.
    Step (3) is justified by the following claim.
    
    \begin{claim}\label{lem:small-cycle}
        If $G_2$ has a \fvsSet of size $k_0$, then it is acyclic or
        contains a cycle of length at most $2k_0$.
    \end{claim}
    \begin{proofsketch}
        Step (2) can produce some vertices with self loops or multiple edges.
        If $G_2$ contains a self loop or a multiple edge,
        these edges form a cycle of length $1$ or $2$ respectively.
        Otherwise $G_2$ is a simple graph with minimum degree $3$,
        so the claim follows by~\cite[Theorem 2.2, Claim]{Downey1995}.
    \end{proofsketch}

    We now describe how all three steps can be performed by a \FO
    program. 
    The conditions for a vertex being removed in step (1) can be tested by simulating \prog',
    so step (1) is FO-definable with  the help of the auxiliary relations of \prog'.
    We emphasise, that only attached trees are being removed and
    the reachability information therefore does not change for the remaining vertices.
    Let $G_1$ denote the graph resulting from step (1).

    Vertices $u$ and $v$ are connected by a new edge in step (2), if they
    both have degree at least 3 in $G_1$ and
    there are vertices $u',v'$, and edges $(u,u')$ and $(v,v')$ for which the following holds.
    \begin{itemize}
        \item All vertices reachable from $u'$ in $G_1-\{(u,u'), (v,v')\}$ are also reachable from $v'$ and vice versa.
        \item All these vertices (including $u',v'$) have degree 2 in $G_1$.
    \end{itemize}
    Again, by simulating \prog' these conditions can be tested in
    FO. If multi-edges are introduced by (2), it suffices to store one
    additional edge (per edge)    in an additional edge relation~$E'$.
    Additionally a loop-edge is inserted if the connected component of a vertex $u$
    contains only vertices of degree 2 and $u$ is the smallest of those vertices.
    Altogether, step (2) is also FO-definable using the auxiliary relations.
    
    The cycle of length at most $2k_0$ in step (3) can be found
    by constructing a canonical breadth-first search tree $T$ of depth $k_0$
    starting from each vertex in parallel.
    More precisely, \prog computes a binary relation $B$ representing the edge relation of $T$ and
    a ternary relation $I$ containing all tuples $(u,v,w)$ where $v$ is on the (unique) path between $u$ and $w$ in $T$.
    We note, that $T$ contains a cycle of length at most $2k_0$ as
    soon as there are vertices $v_1\not=v_2$ with $(u,v_1),(u,v_2)\in B$
    and (a) an edge $(v_1,v_2)$ in $G$ or (b) a vertex $v$ with edges $(v,v_1)$
    and $(v,v_2)$ in $G$.\tzm{The previous sentence is weird. First, $(u,v_1),(u,v_2)\in B$ seems to think that $B$ is the transitive closure of $T$. Second, $v$ is not contained in T and thus this seems not to fit the first part of the sentence. Maybe better: We note, that $T$ represents a cycle of length at most $2k_0$ as
    soon as there are nodes $v_1\not=v_2$ reachable from a node $u$ of $T$
    and (a) an edge $(v_1,v_2)$ in $G$ or (b) a node $v$ with edges $(v,v_1)$
    and $(v,v_2)$ in $G$.} This cycle can then be identified with the
    help of $I$.
    
    The program \prog computes these relations as follows.
    In the $i$-th round, \prog adds $(u,v)$ into $B$, if  there is an
    edge $(u,v)$ in $G$, $u$ is
    currently a leave of $T$ and $v$ is \emph{not} the parent of $u$ in
    $T$.\tsm{Since $G$ is undirected, I think the last condition is needed.}\tzm{Agree.} If $v$ is already in $T$ or a vertex $v$
    would be inserted due to two distinct leaves, a cycle is found and the
    search is stopped. A tuple $(u,v,w)$ is added to $I$ in the $i$-th round when an edge $(v',w)$ gets inserted into $B$,
    such that (a) $(u,v,v') \in I$, or (b) $v=w$ and $(u,v',v') \in I$.
    After $k_0$ rounds $I$ contains all tuples $(u,v,w)$ such that $v$ is on the BFS path between $u$ and $w$.
    
    So, a new search tree child can be computed in $k_0+2$ rounds.
\end{appendixproof}

\begin{toappendix}
 The proof of \Fref{prop:fVSparaTD} is completed by showing that reachability can be maintained under vertex removal. Here, by saying that a vertex is removed from a graph, we mean that all adjacent edges of this vertex are deleted.\footnote{We recall that the
  dynamic complexity setting does not
  allow \emph{real} vertex deletions.} We do not allow to insert edges to a
removed vertex afterwards, since this suffices for the
purpose of \Fref{prop:fVSparaTD}. 

\newnotation{\vin}{v^{\text{in}}}
\newnotation{\vout}{v^{\text{out}}}

  \begin{proposition}
    \label{prop:reach_complex_changes}
    Reachability in directed graphs can be maintained in \DynFO under
    single edge changes and removal of vertices.
  \end{proposition}
  \begin{proof}
    In \cite{DattaKMSZ18} it was shown that there is a dynamic program
    $\prog$ that maintains \reach under single edge changes.  Let
    $G = (V,E)$ be the input graph.  The simple idea is to replace
    each vertex $v$ by two vertices \vin and \vout and an edge
    $(\vin,\vout)$. All in-coming edges of $v$ lead to \vin and all
    out-going edges leave from \vout. This allows to simulate the
    deletion of all edges of $v$ in $G$ by removing just
    $(\vin,\vout)$ in the new graph.

    Technically, this is a bfo-reduction in the sense of
    \cite{DattaKMSZ18} and therefore the proposition basically follows
    from \cite[Proposition 4]{DattaKMSZ18}.
  \end{proof}

\end{toappendix}

\subsection{Kernelisation}\label{subsec:kernel}

Bannach and Tantau~\cite[Theorem 2.3]{BannachT18}
show that the famous meta-theorem ``a problem is fixed parameter tractable
if and only if a kernel for it can be computed in polynomial
time'' can be adapted to connect the \AC-hierarchy with its
parameterised counterpart. In this section we (partially) translate this
relationship to the parameterised, dynamic setting.

A \emph{kernelisation} of a Boolean parameterised query $(Q, \kappa)$ over schema $\schema$ is a self-reduction $K$ from \schema-structures to \schema-structures such that (1) $\inp \in Q$ if and only if $K(\inp) \in Q$, and (2) $|K(\inp)| \leq h(\kappa(\inp))$, for all $\schema$-structures $\inp$ and some fixed computable function $h: \N \rightarrow \N$. The images of a kernelisation $K$ are called \emph{kernels}.  We say that a kernel of $(Q, \kappa)$ can be maintained in some class $\calC$  under some set $\Delta$ of change operations, if the kernels with respect to some kernelisation $K$ can be maintained in $\calC$ under changes from $\Delta$.%

\begin{theoremrep}\label{thm:kernelParaSD}
Let $(Q,\kappa, \Delta)$ be a Boolean parameterised dynamic query of $\schema$-structures.
\begin{enumerate}
 \item If a kernel for $(Q,\kappa)$ can be maintained under $\Delta$ in \DynFOar then $(Q,\kappa,\Delta)$ is in \ParaSD. 
 In addition, if $(Q, \kappa)$ has an explicit parameter and $\Delta = \Delta_\schema \cup \pmone$ then also the converse holds.
 \item If $Q \in \PSPACE$ and a kernel for $(Q,\kappa)$ can be maintained under $\Delta$ in \DynFOar then $(Q,\kappa,\Delta)$ is in \ParaTD.
\end{enumerate}
\end{theoremrep}
\begin{proofsketch}
Towards proving (a), suppose that a kernel of $(Q, \kappa)$ with respect to a
kernelisation $K$ can be maintained under $\Delta$ by a $\DynFOar$-program $\prog$. 
A $\ParaSD$-program~$\prog'$ for $(Q,\kappa,\Delta)$ maintains a kernel for the current input structure by simulating~$\prog$.
The kernel $K(\inp)$ of an input structure $\calI$ is represented by at most $h(\kappa(\inp))$ elements, where $h$ is the function from the second condition of the definition of the kernelisation $K$. Therefore $\prog'$ can check whether $K(\calI) \in Q$ by Lemma~\ref{lem:small} and Proposition~\ref{prop:SDynS}.

For  proving the converse of (a) under the stated assumptions, suppose that $(Q, \kappa)$ has an explicit parameter and that $\Delta = \Delta_\schema \cup \pmone$. 
We construct, from a \ParaSD-program \prog with advice $\pi$ that maintains $(Q,\kappa, \Delta)$, a \DynFOar-program~$\prog'$ that maintains a kernel for $(Q,\kappa)$. 
The idea is to use a standard trick from parameterised complexity, a case distinction between small and large parameters. If the parameter is small enough in comparison to the domain size, $\prog'$ can compute the advice structure of $\prog$ at initialisation time and  can simulate $\prog$ from then on. If the parameter is large, $\prog'$ uses the ``small'' input instance as a trivial kernel.

Towards proving (b), suppose that a kernel of $(Q, \kappa)$ with respect to a kernelisation $K$ can be maintained under $\Delta$ by a $\DynFOar$-program $\prog$, and that $Q \in \PSPACE$. Recall that unlimited (or equivalently exponential) iteration of \FO-formulas captures $\PSPACE$  over ordered structures (see, e.g., \cite[Theorem 10.13]{ImmermanDC}). A \ParaTD-program can  maintain the current kernel $K(\inp)$ by simulating $\prog$. After updating the kernel after a change, it computes the result of $Q$ for $K(\inp)$ by iterating the first-order formulas of the \PSPACE algorithm with a parameterised first-order program. %
Since at most $2^{|K(\inp)|{^{O(1)}}}$ iterations are necessary, it follows that the first-order program only needs a parameterised number of iterations.   
\end{proofsketch}

\begin{appendixproof}
  We make the approach for the converse of (a) more precise next. 
Suppose that the advice $\pi(k)$ can be computed by a Turing machine $M$ with time bound $f(k)$, for some non-decreasing computable function $f$ and all $k \in \N$. The computation of $M$ for a parameter value $k$ can thus be encoded by a binary string of length at most $f(k)^2$ (with some suitable binary encoding).

The program $\prog'$ maintains a kernel of size at most $h(k) \df 2^{f(k)^2}$.
To this end, let $n$ be the size of the domain, and let $k$ be a parameter value. 
The program $\prog'$ first determines the largest number $k'$, for which the computation of $\pi(k')$ is encoded in the binary representation of some element  of the domain $D$. More precisely, if $h(k') \leq  n$, then each binary string $s$ of length $f(k')^2$ can be represented by an element $a_s$ of the domain. Arithmetic on the domain allows to access the bit string encoded by $a_s$ (see, for instance, \cite[Theorem 6.12]{Libkin04}), and to verify whether $s$ indeed represents the computation of $M$ on input $k'$ in a first-order fashion.
From $a_s$, the program $\prog'$ thus extracts the advice structure $\pi(k')$ and stores it in its auxiliary structure.
All this can be done in a first-order fashion at initialisation time.
From the reasoning above we also deduce that $h(k') \leq  n$, but $h(k'+1) > n$.

Afterwards $\prog'$ simulates $\prog$ with the proviso that, if the input structure is $\inp = (\inp',k)$, it simulates $\prog$ with input structure $(\inp',\min(k,k'))$, thus making sure that the parameter never exceeds the value $k'$, for which the advice structure is available. 

Whenever $k>k'$ it holds $n<h(k)$ and $\prog'$ simply outputs $(\inp',k)$.

Otherwise, $\prog'$ output a fixed positive or negative structure, depending on whether $\prog$ accepts or rejects $(\inp',k)$. 
\end{appendixproof}

The assumptions for the proof of the direction $(2) \Rightarrow (1)$ are chosen because they are easy to state and satisfied by many natural parameterised dynamic queries. They can be relaxed though and, as an example, the result also holds for the standard change operations and the non-explicit parameter ``maximal node degree'' for graphs.

We now give an example of an algorithm whose underlying kernelisation can be simulated in \DynFOar.
For a set of points in $\N^d$, for some $d\ge 2$, a \emph{cover} is a set of lines  such that each of the points is on at least one line.
For a fixed dimension $d\geq2$, the problem \pointLC (``PointLineCover'') is defined as follows: %

\paraProblemDescription{\pointLC}{Distinct points $\tpl p_1, \ldots, \tpl p_n \in \N^d$}{$k$}{Is there a cover of the points of size  $k$?}

Each point $\tpl p_i$ with $i \in [n]$ is given by $d$ coordinates $p_i^1, \ldots, p_i^d$ of $n$ bits each.
To encode these numbers, we identify the domain of size $n$ with the set $[n]$ and use $d$ binary relations $X^1, \ldots, X^d$. We let $(i, j) \in X^\ell$ if the $j$-th bit of $p_i^\ell$ is $1$.

A classical kernel~(see e.g. \cite{KratschPR16} or \cite{LangermanM05}) for \pointLC can be obtained by realising that if a line contains at least $k+1$ points then it has to be used in a cover.  Otherwise the points on this line can only be covered by using at least $k+1$ distinct lines. A kernel for an instance can now be constructed by iteratively applying the following rule as long as possible: remove all points that belong to a simple
line that contains at least $k+1$ points and reduce $k$ by $1$. If, in the end, more than $k^2$ points remain, there is no line cover with $k$ lines.

In \cite{BannachT18} it was observed that the above reduction can be performed in parallel, since removing all points of a line removes at most one point from any other line. This immediately yields that \pointLC is in $\ParaTC^0$, since lines with at least $k+1$ points can be identified in $\TC^0$. The problem, however, is not in $\ParaAC^0 = \ParaS$  \cite{BannachT18} due to the bottleneck that collinearity of $n$-bit points cannot be tested in $\AC^0$.

We show that with an oracle for testing whether three points are collinear, a kernel of $\pointLC$ can be actually expressed in \FOar. Since collinearity of three points can be maintained in \DynFOar under bit changes of points, a kernel can be maintained in \DynFOar. Here the allowed changes are to modify single bits of the points $\tpl p_1, \ldots, \tpl p_n$, to enable or disable a point, and to change the number $k$. To allow that points can be enabled or disabled, we add an additional unary relation $P$ to structures that contains $i$ if $\tpl p_i$ is part of the current instance, that is, if it is \emph{enabled}.

\begin{lemmarep}\label{lemma:collinearity}
  Collinearity of three $d$-dimensional points with $n$-bit coordinates can be maintained in $\DynFOar$ under changes of single bits, for each fixed $d \in \N$.
\end{lemmarep}
\begin{appendixproof}
  Three points $\tpl p_1$, $\tpl p_2$, and $\tpl p_3$ with $\tpl p_i = (p^1_i, \ldots, p^d_i)$ are collinear if $p^1_1 \neq p^1_2$,  $p^1_1 \neq p^1_3$  and 
  $\frac{p^j_1 - p^j_2}{p^1_1 - p^1_2} = \frac{p^j_1 - p^j_3}{p^1_1 - p^1_3}$
  for all $j \in \{2, \ldots, d\}$. This is equivalent to 
   $p^j_1p^1_1 - p^j_2p^1_1 - p^j_1 p^1_3 + p^j_2p^1_3 = p^j_1p^1_1 - p^j_3p^1_1 - p^j_1 p^1_2 + p^j_3p^1_2$.
In the case where $p^1_1 = p^1_2 = p^1_3$ holds, analogously formed equations for the second dimension are used, and so one.
  
  In \cite{PatnaikI97} it was shown that the product of $n$-bit numbers can be maintained under single bit changes. Thus, by maintaining the products of all components of $\tpl p_1$, $\tpl p_2$, and $\tpl p_3$, collinearity can be checked since $n$-bit numbers can be added in $\FOar$. 
\end{appendixproof}

\begin{theorem}
    Let \mbox{$\Delta \df \Delta_{\{X_1, \ldots, X_d, P\}}\cup\{\pmone\}$}.
    \begin{enumerate}[(1)]
        \item $(\pointLC, \Delta) \in \ParaSD$
        \item $(\pointLC, \Delta) \in \ParaTD$
    \end{enumerate}
\end{theorem}

\begin{proof}[Proof idea]

  By the previous lemma, a dynamic program can maintain a relation $C$ that contains a triple $(i_1, i_2, i_3)$ if the points $\tpl p_{i_1}, \tpl p_{i_2}, \tpl p_{i_3}$ are collinear, using Lemma \ref{lemma:collinearity}. The statement now follows from Theorem \ref{thm:kernelParaSD} and  the observation that a kernel can be defined in $\FOar$ from $C$.

    If $k \geq \log n$, the input structure \inp itself is a kernel of size at most $f(k)$.  Otherwise, the counting abilities of $\FOar$ (see for example \cite{DenenbergGS86}) can be used to define a kernel.  Since $k < \log n$, the set $L$ of lines with at least $k+1$ enabled points can be defined in \FOar, as well as the number $|L|$ of such lines. Additionally, the set $P$ of enabled points that are not on any line from $L$ is definable, and it can be determined in $\FOar$ whether there are more than $k^2$ of these points. Then the current kernel is defined as follows. If $|L| > k$, or $|L| \leq k$ and $|P| > k^2$, then it outputs a constant no-instance. Otherwise the kernel is the set $P$ with the parameter $k-|L|$.
\end{proof}

\subsection{Dynamic programming}\label{subsec:dynprogramming}

Dynamic programming is a fundamental technique
in algorithm design and as such it has been applied
in the field of parameterised algorithms many
times (e.g., \cite[Section 9]{Niedermeier06}).
A classical parameterised algorithm with dynamic programming shows
$\knapsack\in\FPT$.
\paraProblemDescription{\knapsack}{A set of $n$ items with profits $p_1, \ldots, p_n$
and weights $w_1, \ldots, w_n$, a capacity bound~$B$ and a profit threshold $T$}{
$B$}
{Is there a subset $S \subseteq [n]$ such that
$\sum_{i \in S} p_i \geq T$ and $\sum_{i \in S} w_i \leq B$?}
All numbers are from $\N$ and given as $n$-bit numbers. We choose a similar input encoding as for \pointLC in Subsection \ref{subsec:kernel}:
we identify the domain of size $n$ with the set~$[n]$, encode the profits $p_i$ using a binary relation $P$ such that $(i,j) \in P$ if the $j$-th bit of $p_i$ is $1$, and analogously encode the weights $w_i$ and the numbers $B, T$ by a binary relation $W$ and unary relations $B, T$, respectively.\footnote{We note that this 
restricts the possible weights and profits to numbers bounded by $2^n-1$. 
Larger values can be achieved by  a larger domain, where additionally represented items have profit and weight $0$.}

\begin{propositionrep}
\label{prop:knap}
 $(\knapsack,\Delta_\mtext{KS})\in\ParaSD$.
\end{propositionrep}
Here, $\Delta_\mtext{KS}$ denotes the set of changes that can
arbitrarily replace the profit and the weight of one item, and  set a
number $B$ or $T$ to any value.
\begin{proofsketch}
  The program combines the usual static algorithm with an idea that
  was  used to capture regular languages in \DynFO
  \cite{GeladeMS12}.  Intuitively, it maintains a three-dimensional
  table $A$ such that $A(i,j,b)$ gives the maximum profit one can
  achieve by picking items with overall weight exactly $b$ from
  $\{i,\ldots,j\}$. This table is encoded by a relation $A_\mtext{bit}$ of arity
  four in a straightforward manner. 
\end{proofsketch}

\begin{appendixproof}
We describe a parameterised dynamic program $\prog$ that intuitively
uses a table $A$ with dimensions $[n]\times [n] \times [B]$ to
maintain $(\knapsack,\Delta_\mtext{KS})$, with the intention that
$A(i,j,b)$ gives the maximum profit one can achieve by picking items
with overall weight exactly $b$ from
$[i,j]\df\{i,\ldots,j\}$. Technically, $A$ is encoded by a relation $A_\mtext{bit}$ with the intention that $(i,j,b,m) \in A_\mtext{bit}$ exactly
  if the $m$-th bit of the number $A(i,j,b)$ is $1$, where $i,j,m \in
  [n]$ and $b \in [B_\mtext{max}]$. 

Let $B_\mtext{max}$ be the given upper bound on the parameter value. The advice of $\prog$ consists of the domain $[B_\mtext{max}]$, together with the natural linear order and the \BIT predicate.
The dynamic program maintains a relation $A_\mtext{bit}$ of arity four with the intention that $(i,j,b,m) \in A_\mtext{bit}$ exactly if the $m$-th bit of the number $A(i,j,b)$ is $1$, where $i,j,m \in [n]$ and $b \in [B_\mtext{max}]$.
We present $\prog$ on the basis of $A$ in the following, the translation to $A_\mtext{bit}$ is obvious.

If there is a $b \leq B$ such that $A(1,n,b) \geq T$, then $\prog$ accepts the current input instance. No auxiliary relation needs to be updated under changes of $B$ and $T$, so we only need to sketch how $\prog$ updates an entry $A(i,j,b)$ when the profit $p_\ell$ and the weight $w_\ell$ of some item~$\ell$ is changed.
We assume $\ell \in [i,j]$, as otherwise no update is necessary.
If item $\ell$ shall not be part of the selection for this entry, the largest possible value $P_1$ one can achieve is given as $P_1 \df \max_{b_1 + b_2 = b} A(i,\ell-1,b_1) + A(\ell+1,j,b_2)$.
Otherwise, the largest possible value is $P_2 \df \max_{b_1 + b_2 + w'_\ell = b} A(i,\ell-1,b_1) + A(\ell+1,j,b_2) + p'_\ell$, where $p'_\ell$ is the changed profit and $w'_\ell$ is the changed weight of $\ell$.
The maximum of $P_1$ and $P_2$ becomes the updated value of $A(i,j,b)$.

The update can be expressed in \FOar since comparison and addition of two $n$-bit numbers is \FOar-expressible and the update formulas can existentially quantify elements from $[B_\mtext{max}]$ and translate between these elements and their encoding as $n$-bit numbers via the $\BIT$ predicate.
\end{appendixproof}

\begin{toappendix}
Before we turn to the proof of \Fref{prop:fVSparaSD}, we describe a technique from dynamic complexity that can be adapted to the parameterised setting.

If a query is maintained by a dynamic program, the program needs to be able to deal with arbitrarily long change sequences, that is, to run ``forever''. The \emph{muddling technique} from \cite{DattaMSVZ19} allows to soften this requirement under certain circumstances. %
More precisely, it allows to show the existence of a dynamic program by showing the existence of a dynamic program that can handle a \emph{bounded number of change steps}, starting from an arbitrary input structure and suitable auxiliary relations.

We formalise the muddling technique for the parameterised dynamic setting next. For a parameterised dynamic query $(Q,\kappa,\Delta)$, we call $\Delta$ \emph{gradual}, if a single change operation from $\Delta$ affects at most $d$ elements of the domain, for some $d\in\N$, and increases the parameter of a structure at most by one.
We say that a parameterised dynamic query $(Q,\kappa,\Delta)$ is \emph{short-term maintainable}, if $\Delta$ is gradual and there are non-decreasing computable functions $f,g,h$, a $(f,g)$-parameterised first-order program $(\FOprog,\pi)$ and a $(h,1)$-parameterised dynamic program $\prog$ (with advice but without iteration) that for any input structure $\inp$ with parameter $k \df \kappa(\inp)$ maintains $Q$ for $g(k+1)$ change steps $\alpha_1, \ldots, \alpha_{g(k+1)}$, starting from the state $(\inp,\FOprog(\inp))$.
That is, if $\FOprog$ needs $g(k)$ iterative steps to compute auxiliary relations for an arbitrary initial structure, then $\prog$ needs to maintain $Q$ for $g(k+1)$ change steps.

We emphasise three crucial differences between this definition and our ``standard'' maintainability: the computation does \emph{not} start from an empty structure, but from an arbitrary structure $\inp$. Therefore the need for initial auxiliary relations arises, however their computation can take only a parameterised number of rounds. Finally, the query needs to be maintained only for a slightly larger number of change steps.

This notion of maintainability is a variant of the notion of $(\calC,f)$-maintainability as defined in \cite{DattaMSVZ19}, which asks that a query is maintained for $f(n)$ many change steps, where $n$ is the size of the domain, after an initialisation that can be computed with complexity $\calC$.

The application of the muddling technique requires a technical condition, that ensures that the query under consideration does not crucially depend on ``isolated'' elements. 
Let $\adom(\db)$ denote the \emph{active domain} of the structure $\db$, that is, the set of elements of $\db$ that are used in some tuple or as some constant in $\db$.
We call a query $\query$ \emph{almost domain independent}, c.f.~\cite{DattaKMSZ18, DattaMSVZ19}, if there is a $c \in \N$ such that for every structure $\db$ with domain $D$ and every set $A \subseteq D \setminus \adom(\db)$ with $|A| \geq c$ it holds  $\restrict{\query(\db)}{(\adom(\db) \cup A)} = \query(\restrict{\db}{(\adom(\db) \cup A)})$.
Intuitively, this means that if the structure has at least $c$ domain elements that do not appear in any relation, then the query result does not depend on the exact number of such elements.
  
\begin{lemma}
\label{lem:muddling}
Every short-term maintainable dynamic parameterised query $(Q,\kappa,\Delta)$ with almost domain independent $Q$ is in $\ParaSD$.
\end{lemma}

The high-level proof idea is as follows. Let $\prog$ and $\FOprog$ be as above and let, for every $t>0$, $\inp_t$ denote the input instance at time $t$ with parameter $k_t \df \kappa(\inp_t)$ . Here, each change operation represents a time step.
Then $Q$ can be maintained in $\ParaSD$ by a combined program $\prog'$ as follows. 
We view $\prog'$ as a parallel composition of several copies of some dynamic program, which we call \emph{threads}.
At time $t$,  $\prog'$ starts a thread computing the initial auxiliary relations for  $\inp_{t}$. This thread applies the respective formulas of $\FOprog$ twice per change step and thus completes this computation at time $t+\frac{g(k_t)}{2}$. 
In the next $\frac{g(k_t)}{2}$ rounds, $\prog'$ applies the $g(k_t)$ change operations that happen(ed) between time $t$ and $t+g(k_t)$, two at a time, by simulating $\prog$ twice. At time $t+g(k_t)$, $\prog'$ has computed  $Q(\inp_{t+g(k_t)})$. It can further maintain $Q$ until (including) time $t+g(k_t+1)$. At time $t+1+g(k_t+1)$ or earlier, the thread that starts at time $t+1$ takes over.

\begin{proof}[Proof sketch]
The proof follows the lines of \cite[Theorem~4.2]{DattaMSVZ19}.
Let $(\FOprog,\pi)$ be the $(f,g)$-parameterised first-order program and $\prog$ the dynamic program that witnesses that the
almost domain independent (with constant $c$) parameterised dynamic query
$(Q,\kappa,\Delta)$, where $\Delta$ is gradual with constant $d$, is short-term maintainable. We describe
a program $\prog'$ with advice $\pi$ that maintains $(Q,\kappa,\Delta)$ in
$\ParaSD$. Let $\kmax$ be the given upper bound on the parameter value.

For the first $g(\kmax)$ time steps, $\prog'$ can maintain $Q$ over a domain of size $dg(\kmax)+c$ with the help of a suitable advice structure along the lines of Lemma~\ref{lem:small}.   Since $Q$ is almost domain independent, the query result is the same over this domain  and the full domain. 

Afterwards,  $\prog'$ has at most $g(\kmax)$ threads and administrates
them in a round robin fashion. The thread starting at time $t$  is responsible for delivering the query answer from time $t+g(k_t)$ to $t+g(k_t+1)$, where $k_t$ is $\kappa(\inp_t)$. Since $\Delta$ is gradual, for each time point at least one thread is responsible. 
No conflicts can occur if there is more than one responsible thread for some time point, as all of them yield the same answer.

Let $T$ be the thread that starts at time $t$. It works in three phases. The first phase lasts from time point $t$ until $t+\frac{g(k_t)}{2}$ and in this phase $T$ simulates $\FOprog = (\Psi, \varphi)$, by applying the first-order formulas from $\Psi$ two times for every time step. The change operations that occur during this time are stored but not directly processed by $T$.
In the second phase, which lasts from time point $t+\frac{g(k_t)}{2}$ until $t+g(k_t)$, $T$ simulates $\prog$ and applies the changes that occurred from $t$ until $t+g(k_t)$, two at a time. The third phase starts at time $t+g(k_t)$ and may last until time $t+g(k_t+1)$. In this phase, $T$ still maintains $Q$ with the help of $\prog$ and yields the query result.

Thread $T$  maintains the following
relations:
\begin{itemize}
\item
a counter in order to know in which phase the thread is,
\item
relations $A_\delta$ for every $\delta \in \Delta$ that store the
changes that
have been applied to the input during the two phases,
\item
its own versions of the input relations,
\item
its own versions of the auxiliary relations of $\prog$. 
\end{itemize}
The separate auxiliary relations of each of the $g(\kmax)$ threads can be combined to auxiliary relations for the dynamic program $\prog'$ by having the  thread number as a new component to each tuple.
\end{proof}

 \end{toappendix}

\subsection{Iterative compression}\label{subsec:iterative-compression}

The iterative compression method (introduced in \cite{ReedSV04}, see also~\cite[Section 11.3]{Niedermeier06}) 
is used to obtain fixed parameter tractable algorithms for minimisation problems 
which are parameterised by the solution size.
It can roughly be described as follows: First, a trivial solution is
computed for a very small fraction of the input instance. Afterwards,
the fraction is continuously increased and each time a straightforwardly updated (but maybe too big) 
            solution is constructed and improved (``compressed'')
            afterwards (if necessary), until the input instance is completed and a valid solution is constructed.
We illustrate the transfer of this technique to the dynamic setting
with  \vCover. First we describe intuitively, how the static algorithm described in
\cite[Subsection 11.3.2]{Niedermeier06} can be adapted to the dynamic setting.

Let $G=(V,E)$ and $G'=(V,E')$ be two input graphs, where $G'$ results
from $G$ by inserting one edge $e=(u,v)$.  
Let us assume  that $C_0$ is an optimal vertex cover for $G$ of size $k$.
The set $C=C_0 \cup \{ u\}$ of size $k+1$  is trivially a vertex cover for $G'$, but
the optimal one $C'$ might have size $k$.  The crucial observation is that
if $C'=Z\cup Z'$ has size $k$, for a subset $Z$ of $C$ and a set $Z'$ disjoint from
$C$, then $Z'$ must consist of all neighbours of vertices in $C-Z$ that are
not in $Z$.  By a combination of colour coding with an adaptation
of a technique from \cite{DattaMSVZ19} for the parameterised setting,  a
dynamic program with advice (for the universal colouring family) can 
basically try out all subsets of $C$ for $Z$.  

\begin{propositionrep} \label{prop:fVSparaSD}
    $(\vCover, \Delta_E \cup \pmone) \in \ParaS$ by a compression-based dynamic program.
\end{propositionrep}
\begin{appendixproof}
  The static compression algorithm considers all possible intersections
between a better vertex cover $C'$ and $C$ (that is, all subsets $Z$ of $C$ of size at most  $k$)
and checks whether one of these intersections can be extended to a vertex cover for $G_{i+1}$ of size $k$. 
Thus, in a compression step, the algorithm has to solve the
following problem, at most $2^{k+1}-1$  times:

\paraProblemDescription{\disVCover}
    {An undirected graph $G=(V,E)$, a vertex cover $C$ for $G$, such
      that $|C|=k+1$ and $Z \subseteq C$}
    {$k$}{Is there a $C' \subseteq V$ such that $C'$ is a vertex cover for $G$, $|C'|=k$ and $C \cap C' = Z$?}

With the help of the following observation, solving \disVCover is easy. Let $G[X]$ denote the subgraph of $G$ induced by $X \subseteq V$.
Since $C$ is a vertex cover for $G$, the set $C \setminus Z$ is a vertex cover for $G[V \setminus Z]$
and every edge of $G[V \setminus Z]$ has one of its endpoints
in~$C \setminus Z$. 
Because $C'$ cannot contain any vertex from $C \setminus Z$, it needs to include all neighbours of these vertices.
So, the only candidate for $C'$ is the union of $Z$
and the set of all neighbours of vertices in $C \setminus Z$. 
It is easy to compute this set and to test whether it satisfies all conditions.
Altogether this yields an \FPT-algorithm for  \vCover.
As this algorithm obtains its solution by continuously adding edges,
its technique is amenable for the dynamic setting.

  We construct a parameterised dynamic program \prog that maintains
  \vCover based on compression. For simplicity, we first assume that
  the input graph has at all times a vertex cover of size at most
  $2k$ (where $k$ always denotes the current parameter value). In this case, the compression technique described above
  can be applied almost immediately. If an edge is modified, the
  dynamic program proceeds essentially as described above.%

  More precisely, the program \prog maintains a minimal vertex cover as well as its size, both stored in unary relations. The size can be used to answer the query and to handle changes of the parameter.

  Edge changes are handled as follows. Let $G$ and $G'$ be the old and the modified graph, respectively, and let $C$ be a minimal vertex cover for $G$ of size at most $2k$. 
If an edge $(u, v)$ is inserted, a new trivial cover $C_\text{t}$ for $G'$ can be chosen as $C_\text{t} \df C \cup \{u\}$, and we choose $C_\text{t} \df C$ in case of an edge deletion.
However, $C_\text{t}$ might not be a minimal solution. 
To compress $C_\text{t}$, the program first expresses every $Z \subseteq C_\text{t}$ with the help of colour-coding. Since $|C_\text{t}| \leq 2k+1$, an $(n, 2k+1, 2k+1)$-universal colouring family contains a colouring that maps each vertex in $C_\text{t}$ to a colour that is unique among vertices in $C_\text{t}$. 
Additionally, the advice structure of $\prog$ contains an element for every subset of the $2k+1$ colours, and a relation that connects these elements with the colours contained in the represented set.

Then \prog solves \disVCover for each $Z$ in parallel. Recall that the only solution for \disVCover for fixed $Z$ is $C' = Z \cup N(C_\text{t} \setminus Z)$
  where $N(X)$ is the set of all neighbours of vertices in $X$. 
The sets $C'$ are clearly \FO-definable, and $\prog$ can check whether $C \cap C' = Z$ holds. 
Again with the help of colour-coding, \prog can also check if $C'$ has size $|C_\text{t}|-1$. More precisely, \prog checks if an $(n, |C_\text{t}|-1, |C_\text{t}|-1)$-universal colouring family contains a colouring that colours each vertices in $C'$ uniquely.

  The program chooses the lexicographically smallest $C'$ if such a set exists, and otherwise selects $C_\text{t}$ as the new minimal vertex cover for the graph $G'$.
  
With the help of the muddling technique from Lemma~\ref{lem:muddling} we drop the assumption that there is a vertex cover of size $2k$ at all times, and show that $(\vCover, \Delta_E \cup \pmone)$ %
is short-term maintainable.
    More precisely, we show that there is a $(2^{2k},k-1)$-parameterised \FO-program that uses $k-1$ iterations to compute a vertex cover of size at most $2k$ if such a vertex cover exists, and rejects otherwise.
Then, the dynamic program as constructed above can maintain \vCover for $k$ change steps, as long as the minimal vertex cover does not exceed the size bound $2k$.
If that happens, either already for the initial graph or during the $k$ change steps, then the remaining up to $k$ changes cannot transform the input instance into a graph with a vertex cover of size at most $k$. So, during this time the dynamic program can trivially answer ``no''.

    It remains to give the details for the $(2^{2k},k-1)$-parameterised \FO-program that initialises the auxiliary data. This program basically constructs the search tree from Example~\ref{ex:vc-tree-dyn}, appending four levels in the first iteration and two levels in every subsequent iteration.
So, after $k-1$ iterations a search tree of depth $2k$ is available. 
Of course, the program also computes the necessary advice for this search tree and for the dynamic program $\prog$.
\end{appendixproof}

\section{Conclusion}
\label{section:conclusion}
In this work we started to investigate dynamic complexity from a parameterised algorithms point of view.
Besides the definition of the framework, we
explored how well-known techniques from parameterised algorithms
translate to our setting. Kernelisation and colour-coding worked quite well
for both settings. Search-tree based techniques translated well
to the setting with parameterised time and were more challenging for
parameterised space. On the other hand, dynamic programming (with superpolynomial parameter
values) seems better suited for parameterised space. The
compression-based program for \vCover translates, in principle, also
to   \ParaTD but the handling of instances with large minimal vertex
cover basically requires an additional implementation of some other
method and therefore makes this approach a bit pointless.
We also considered \emph{greedy localisation} and algorithms for
structures with bounded tree-width, but did not find any meaningful
applications in the dynamic setting, as discussed in the appendix.
\begin{toappendix}
  We briefly discuss two methods
which we were unable to transfer into the dynamic
setting.

The greedy localisation method~\cite[Section 11.4]{Niedermeier06}
is applied to maximisation problems that are
parameterised by the solution size. In this method
a maximal solution $S$ is computed
using a greedy algorithm and then
the information given by $S$ is used in order to localise
(and thus reduce) the
search space around $S$, so that an optimal solution can be found by brute force. Unfortunately the steps of
a greedy algorithm are usually inherently sequential,
thus maintaining the result of a greedy algorithm in
the dynamic parameterised classes defined here
seems very difficult. Things can get even more
complicated when solving a problem with greedy
localisation requires the repetitive application of 
the greedy algorithm. All these complications have
kept us from transferring the greedy localisation
method to the dynamic parameterised setting.

Courcelle's Theorem \cite{Courcelle90} implies that for every monadic second-order (\MSO) formula $\psi$ there is an \FPT algorithm that decides whether an input graph $G$ satisfies $\psi$, with parameter being the \emph{treewidth} of $G$.
Although each \MSO-defined graph property is in \DynFO for graphs with bounded treewidth \cite{DattaMSVZ19}, a corresponding result for the parameterised dynamic setting is unknown.
There are at least two bottlenecks. Firstly, the result of \cite{DattaMSVZ19} relies on the fact that tree decompositions for graphs of bounded treewidth can be computed in \LOGSPACE \cite{ElberfeldJT10}, which is not known for the parameterised counterpart \cite{BannachT16}.
Secondly, the update formulas from \cite{DattaMSVZ19} quantify over structures of polynomial size in the input length, where the treewidth determines the degree of the polynomial. It is unclear how the size of these structures can be restricted in an ``\FPT way''.
\cite{BannachT16}
\end{toappendix}

Particular open questions are whether \cString or \fVS can be maintained with
parameterised space and %
whether \ParaSTD is more expressive than \ParaSD.

\bibliography{bibliography}

\begin{thebibliography}{10}

\bibitem{AlmanMW17}
Josh Alman, Matthias Mnich, and Virginia~Vassilevska Williams.
\newblock Dynamic {P}arameterized {P}roblems and {A}lgorithms.
\newblock In Ioannis Chatzigiannakis, Piotr Indyk, Fabian Kuhn, and Anca
  Muscholl, editors, {\em 44th International Colloquium on Automata, Languages,
  and Programming, {ICALP} 2017}, volume~80 of {\em LIPIcs}, pages 41:1--41:16.
  Schloss Dagstuhl - Leibniz-Zentrum fuer Informatik, 2017.
\newblock \href {http://dx.doi.org/10.4230/LIPIcs.ICALP.2017.41}
  {\path{doi:10.4230/LIPIcs.ICALP.2017.41}}.

\bibitem{AlonYZ95}
Noga Alon, Raphael Yuster, and Uri Zwick.
\newblock Color-coding.
\newblock {\em J. {ACM}}, 42(4):844--856, 1995.
\newblock \href {http://dx.doi.org/10.1145/210332.210337}
  {\path{doi:10.1145/210332.210337}}.

\bibitem{BannachST15}
Max Bannach, Christoph Stockhusen, and Till Tantau.
\newblock {F}ast {P}arallel {F}ixed-{P}arameter {A}lgorithms via {C}olor
  {C}oding.
\newblock In {\em 10th International Symposium on Parameterized and Exact
  Computation, {IPEC} 2015}, pages 224--235. Schloss Dagstuhl - Leibniz-Zentrum
  fuer Informatik, 2015.
\newblock \href {http://dx.doi.org/10.4230/LIPIcs.IPEC.2015.224}
  {\path{doi:10.4230/LIPIcs.IPEC.2015.224}}.

\bibitem{BannachT16}
Max Bannach and Till Tantau.
\newblock Parallel multivariate meta-theorems.
\newblock In Jiong Guo and Danny Hermelin, editors, {\em 11th International
  Symposium on Parameterized and Exact Computation, {IPEC} 2016}, volume~63 of
  {\em LIPIcs}, pages 4:1--4:17. Schloss Dagstuhl - Leibniz-Zentrum fuer
  Informatik, 2016.
\newblock \href {http://dx.doi.org/10.4230/LIPIcs.IPEC.2016.4}
  {\path{doi:10.4230/LIPIcs.IPEC.2016.4}}.

\bibitem{BannachT18b}
Max Bannach and Till Tantau.
\newblock Computing hitting set kernels by {AC}{\^{}}0-circuits.
\newblock In Rolf Niedermeier and Brigitte Vall{\'{e}}e, editors, {\em 35th
  Symposium on Theoretical Aspects of Computer Science, {STACS} 2018},
  volume~96 of {\em LIPIcs}, pages 9:1--9:14. Schloss Dagstuhl -
  Leibniz-Zentrum fuer Informatik, 2018.
\newblock \href {http://dx.doi.org/10.4230/LIPIcs.STACS.2018.9}
  {\path{doi:10.4230/LIPIcs.STACS.2018.9}}.

\bibitem{BannachT18}
Max Bannach and Till Tantau.
\newblock Computing {K}ernels in {P}arallel: {L}ower and {U}pper {B}ounds.
\newblock In Christophe Paul and Michal Pilipczuk, editors, {\em 13th
  International Symposium on Parameterized and Exact Computation,{IPEC} 2018},
  pages 13:1--13:14. Schloss Dagstuhl - Leibniz-Zentrum fuer Informatik, 2018.
\newblock \href {http://dx.doi.org/10.4230/LIPIcs.IPEC.2018.13}
  {\path{doi:10.4230/LIPIcs.IPEC.2018.13}}.

\bibitem{BannachT19}
Max Bannach and Till Tantau.
\newblock On the descriptive complexity of color coding.
\newblock In Rolf Niedermeier and Christophe Paul, editors, {\em 36th
  International Symposium on Theoretical Aspects of Computer Science, {STACS}
  2019}, volume 126 of {\em LIPIcs}, pages 11:1--11:16. Schloss Dagstuhl -
  Leibniz-Zentrum fuer Informatik, 2019.
\newblock \href {http://dx.doi.org/10.4230/LIPIcs.STACS.2019.11}
  {\path{doi:10.4230/LIPIcs.STACS.2019.11}}.

\bibitem{BarringtonIS90}
David A.~Mix Barrington, Neil Immerman, and Howard Straubing.
\newblock On uniformity within {NC{\({^1}\)}}.
\newblock {\em J. Comput. Syst. Sci.}, 41(3):274--306, 1990.
\newblock \href {http://dx.doi.org/10.1016/0022-0000(90)90022-D}
  {\path{doi:10.1016/0022-0000(90)90022-D}}.

\bibitem{BockenhauerBRR18reopt}
Hans{-}Joachim B{\"{o}}ckenhauer, Elisabet Burjons, Martin Raszyk, and Peter
  Rossmanith.
\newblock Reoptimization of {P}arameterized {P}roblems, 2018.
\newblock \href {http://arxiv.org/abs/1809.10578} {\path{arXiv:1809.10578}}.

\bibitem{CesatiI98}
Marco Cesati and Miriam~Di Ianni.
\newblock Parameterized parallel complexity.
\newblock In David~J. Pritchard and Jeff Reeve, editors, {\em Euro-Par '98
  Parallel Processing, 4th International Euro-Par Conference, Proceedings},
  volume 1470 of {\em Lecture Notes in Computer Science}, pages 892--896.
  Springer, 1998.
\newblock \href {http://dx.doi.org/10.1007/BFb0057945}
  {\path{doi:10.1007/BFb0057945}}.

\bibitem{ChenF16}
Yijia Chen and J{\"{o}}rg Flum.
\newblock Some {L}ower {B}ounds in {P}arameterized {AC}{\^{}}0.
\newblock In Piotr Faliszewski, Anca Muscholl, and Rolf Niedermeier, editors,
  {\em 41st International Symposium on Mathematical Foundations of Computer
  Science, {MFCS} 2016}, volume~58 of {\em LIPIcs}, pages 27:1--27:14. Schloss
  Dagstuhl - Leibniz-Zentrum fuer Informatik, 2016.
\newblock \href {http://dx.doi.org/10.4230/LIPIcs.MFCS.2016.27}
  {\path{doi:10.4230/LIPIcs.MFCS.2016.27}}.

\bibitem{ChenFH17}
Yijia Chen, J{\"{o}}rg Flum, and Xuangui Huang.
\newblock Slicewise {D}efinability in {F}irst-{O}rder {L}ogic with {B}ounded
  {Q}uantifier rank.
\newblock In Valentin Goranko and Mads Dam, editors, {\em 26th {EACSL} Annual
  Conference on Computer Science Logic, {CSL} 2017}, volume~82 of {\em LIPIcs},
  pages 19:1--19:16. Schloss Dagstuhl - Leibniz-Zentrum fuer Informatik, 2017.
\newblock \href {http://dx.doi.org/10.4230/LIPIcs.CSL.2017.19}
  {\path{doi:10.4230/LIPIcs.CSL.2017.19}}.

\bibitem{Courcelle90}
Bruno Courcelle.
\newblock The monadic second-order logic of graphs. {I}. recognizable sets of
  finite graphs.
\newblock {\em Inf. Comput.}, 85(1):12--75, 1990.
\newblock \href {http://dx.doi.org/10.1016/0890-5401(90)90043-H}
  {\path{doi:10.1016/0890-5401(90)90043-H}}.

\bibitem{DattaKMSZ18}
Samir Datta, Raghav Kulkarni, Anish Mukherjee, Thomas Schwentick, and Thomas
  Zeume.
\newblock Reachability is in {DynFO}.
\newblock {\em J. ACM}, 65(5):33:1--33:24, 2018.
\newblock \href {http://dx.doi.org/10.1145/3212685}
  {\path{doi:10.1145/3212685}}.

\bibitem{DattaMSVZ19}
Samir Datta, Anish Mukherjee, Thomas Schwentick, Nils Vortmeier, and Thomas
  Zeume.
\newblock A strategy for {D}ynamic {P}rograms: {S}tart over and {M}uddle
  through.
\newblock {\em Logical Methods in Computer Science}, 15(2), 2019.
\newblock \href {http://dx.doi.org/10.23638/LMCS-15(2:12)2019}
  {\path{doi:10.23638/LMCS-15(2:12)2019}}.

\bibitem{DenenbergGS86}
Larry Denenberg, Yuri Gurevich, and Saharon Shelah.
\newblock Definability by constant-depth polynomial-size circuits.
\newblock {\em Information and Control}, 70(2/3):216--240, 1986.
\newblock \href {http://dx.doi.org/10.1016/S0019-9958(86)80006-7}
  {\path{doi:10.1016/S0019-9958(86)80006-7}}.

\bibitem{DongST95}
Guozhu Dong, Jianwen Su, and Rodney~W. Topor.
\newblock Nonrecursive {I}ncremental {E}valuation of {D}atalog {Q}ueries.
\newblock {\em Ann. Math. Artif. Intell.}, 14(2-4):187--223, 1995.
\newblock \href {http://dx.doi.org/10.1007/BF01530820}
  {\path{doi:10.1007/BF01530820}}.

\bibitem{downey2014dynamic}
Rodney~G. Downey, Judith Egan, Michael~R Fellows, Frances~A Rosamond, and Peter
  Shaw.
\newblock Dynamic {D}ominating set and {T}urbo-{C}harging {G}reedy
  {H}euristics.
\newblock {\em Tsinghua Science and Technology}, 19(4):329--337, 2014.
\newblock \href {http://dx.doi.org/10.1109/TST.2014.6867515}
  {\path{doi:10.1109/TST.2014.6867515}}.

\bibitem{DowneyF95}
Rodney~G. Downey and Michael~R. Fellows.
\newblock Fixed-{P}arameter {T}ractability and {C}ompleteness {I:} basic
  results.
\newblock {\em {SIAM} J. Comput.}, 24(4):873--921, 1995.
\newblock \href {http://dx.doi.org/10.1137/S0097539792228228}
  {\path{doi:10.1137/S0097539792228228}}.

\bibitem{Downey1995}
Rodney~G. Downey and Michael~R. Fellows.
\newblock Parameterized {C}omputational {F}easibility.
\newblock In {\em Feasible mathematics II}, pages 219--244. Springer, 1995.
\newblock \href {http://dx.doi.org/10.1007/978-1-4612-2566-9_7}
  {\path{doi:10.1007/978-1-4612-2566-9_7}}.

\bibitem{ElberfeldJT10}
Michael Elberfeld, Andreas Jakoby, and Till Tantau.
\newblock Logspace versions of the theorems of {B}odlaender and {C}ourcelle.
\newblock In {\em 51th Annual {IEEE} Symposium on Foundations of Computer
  Science, {FOCS} 2010}, pages 143--152. {IEEE} Computer Society, 2010.
\newblock \href {http://dx.doi.org/10.1109/FOCS.2010.21}
  {\path{doi:10.1109/FOCS.2010.21}}.

\bibitem{ElberfeldST15}
Michael Elberfeld, Christoph Stockhusen, and Till Tantau.
\newblock On the {S}pace and {C}ircuit complexity of {P}arameterized
  {P}roblems: {C}lasses and {C}ompleteness.
\newblock {\em Algorithmica}, 71(3):661--701, 2015.
\newblock \href {http://dx.doi.org/10.1007/s00453-014-9944-y}
  {\path{doi:10.1007/s00453-014-9944-y}}.

\bibitem{FlumG03}
J{\"{o}}rg Flum and Martin Grohe.
\newblock Describing {P}arameterized {C}omplexity {C}lasses.
\newblock {\em Inf. Comput.}, 187(2):291--319, 2003.
\newblock \href {http://dx.doi.org/10.1016/S0890-5401(03)00161-5}
  {\path{doi:10.1016/S0890-5401(03)00161-5}}.

\bibitem{FurstSS84}
Merrick~L. Furst, James~B. Saxe, and Michael Sipser.
\newblock Parity, circuits, and the polynomial-time hierarchy.
\newblock {\em Mathematical Systems Theory}, 17(1):13--27, 1984.
\newblock \href {http://dx.doi.org/10.1007/BF01744431}
  {\path{doi:10.1007/BF01744431}}.

\bibitem{GeladeMS12}
Wouter Gelade, Marcel Marquardt, and Thomas Schwentick.
\newblock The dynamic complexity of formal languages.
\newblock {\em ACM Trans. Comput. Log.}, 13(3):19, 2012.
\newblock \href {http://dx.doi.org/10.1145/2287718.2287719}
  {\path{doi:10.1145/2287718.2287719}}.

\bibitem{HartungN13}
Sepp Hartung and Rolf Niedermeier.
\newblock Incremental {L}ist {C}oloring of {G}raphs, {P}arameterized by
  {C}onservation.
\newblock {\em Theor. Comput. Sci.}, 494:86--98, 2013.
\newblock \href {http://dx.doi.org/10.1016/j.tcs.2012.12.049}
  {\path{doi:10.1016/j.tcs.2012.12.049}}.

\bibitem{ImmermanDC}
Neil Immerman.
\newblock {\em Descriptive complexity}.
\newblock Graduate texts in computer science. Springer, 1999.
\newblock \href {http://dx.doi.org/10.1007/978-1-4612-0539-5}
  {\path{doi:10.1007/978-1-4612-0539-5}}.

\bibitem{KratschPR16}
Stefan Kratsch, Geevarghese Philip, and Saurabh Ray.
\newblock Point {L}ine {C}over: The {E}asy {K}ernel is {E}ssentially {T}ight.
\newblock {\em {ACM} Trans. Algorithms}, 12(3):40:1--40:16, 2016.
\newblock \href {http://dx.doi.org/10.1145/2832912}
  {\path{doi:10.1145/2832912}}.

\bibitem{LangermanM05}
Stefan Langerman and Pat Morin.
\newblock Covering things with things.
\newblock {\em Discrete {\&} Computational Geometry}, 33(4):717--729, 2005.
\newblock \href {http://dx.doi.org/10.1007/s00454-004-1108-4}
  {\path{doi:10.1007/s00454-004-1108-4}}.

\bibitem{Libkin04}
Leonid Libkin.
\newblock {\em Elements of Finite Model Theory}.
\newblock Springer, 2004.
\newblock \href {http://dx.doi.org/10.1007/978-3-662-07003-1}
  {\path{doi:10.1007/978-3-662-07003-1}}.

\bibitem{MansM17}
Bernard Mans and Luke Mathieson.
\newblock Incremental {P}roblems in the {P}arameterized {C}omplexity {S}etting.
\newblock {\em Theory Comput. Syst.}, 60(1):3--19, 2017.
\newblock \href {http://dx.doi.org/10.1007/s00224-016-9729-6}
  {\path{doi:10.1007/s00224-016-9729-6}}.

\bibitem{Niedermeier06}
Rolf Niedermeier.
\newblock {\em {I}nvitation to {F}ixed-{P}arameter {A}lgorithms}.
\newblock Number~31 in Oxford Lecture Series in Mathematics and its
  Applications. Oxford University Press, 2006.
\newblock \href {http://dx.doi.org/10.1093/acprof:oso/9780198566076.001.0001}
  {\path{doi:10.1093/acprof:oso/9780198566076.001.0001}}.

\bibitem{PatnaikI97}
Sushant Patnaik and Neil Immerman.
\newblock {Dyn-FO}: A {P}arallel, {D}ynamic {C}omplexity {C}lass.
\newblock {\em J. Comput. Syst. Sci.}, 55(2):199--209, 1997.
\newblock \href {http://dx.doi.org/10.1006/jcss.1997.1520}
  {\path{doi:10.1006/jcss.1997.1520}}.

\bibitem{ReedSV04}
Bruce~A. Reed, Kaleigh Smith, and Adrian Vetta.
\newblock Finding {O}dd {C}ycle {T}ransversals.
\newblock {\em Oper. Res. Lett.}, 32(4):299--301, 2004.
\newblock \href {http://dx.doi.org/10.1016/j.orl.2003.10.009}
  {\path{doi:10.1016/j.orl.2003.10.009}}.

\bibitem{SchwentickZ16}
Thomas Schwentick and Thomas Zeume.
\newblock Dynamic {C}omplexity: {R}ecent {U}pdates.
\newblock {\em {SIGLOG} News}, 3(2):30--52, 2016.
\newblock URL: \url{http://doi.acm.org/10.1145/2948896.2948899}.

\end{thebibliography}

\end{document}